\documentclass{article}
\usepackage{amsmath}
\usepackage{authblk}
\usepackage{url}
\usepackage{graphics}
\usepackage{graphicx}
\usepackage[margin=2.5cm]{geometry}
\usepackage{amsfonts}
\usepackage{datetime}
\usepackage{epstopdf}
\usepackage{cite}
\usepackage{dsfont}
\usepackage{amssymb}
\usepackage{amsthm}
\usepackage{braket}
\usepackage{cryptocode, multido}
\usepackage{tikz}
\usepackage{bm}
\usepackage{xcolor}
\usepackage{booktabs}
\usepackage{bbold}
\usepackage{subcaption}
\usepackage{float, multirow, bigstrut, colortbl, hhline}
\newtheoremstyle{theorem}
	{6pt}
	{}
	{\itshape}
	{}
	{\bfseries}
	{:}
	{.5em}
	{}
\theoremstyle{theorem}
\newtheorem{theorem}{Theorem}[section]
\newtheorem{lemma}{Lemma}[section]
\newtheorem{corollary}{Corollary}[section]

\newtheoremstyle{definition}
	{6pt}
	{}
	{}
	{}
	{\bfseries}
	{:}
	{.5em}
	{}
\theoremstyle{definition}
\newtheorem{definition}{Definition}[section]
\newtheorem{remark}{Remark}[section]
\newtheorem{example}{Example}[section]

\DeclareMathAlphabet\mathbfcal{OMS}{cmsy}{b}{n}

\renewcommand{\sec}{Sec$.$~}
\newcommand{\app}{Appendix~}
\newcommand{\fig}{Fig$.$~}
\newcommand{\prog}{Prog$.$~}

\newcommand{\lem}{Lemma~}

\newcommand{\defn}{Def$.$~}
\newcommand{\rem}{Remark~}
\newcommand{\ex}{Ex$.$~}
\newcommand{\nl}{~\\}

\newcommand{\Q}{\mathcal{Q}}
\newcommand{\cB}{\mathcal{B}}
\newcommand{\cK}{\mathcal{K}}

\renewcommand{\P}{\mathcal{P}}
\newcommand{\cR}{\mathcal{R}}
\newcommand{\cV}{\mathcal{V}}
\newcommand{\Qsub}{\tilde{\Q}}

\newcommand{\X}{\mathcal{X}}
\newcommand{\Y}{\mathcal{Y}}
\newcommand{\A}{\mathcal{A}}
\newcommand{\B}{\mathcal{B}}
\newcommand{\ABXY}{\A\B\X\Y}
\newcommand{\C}{\mathcal{C}}
\newcommand{\R}{\mathbb R}

\newcommand{\G}{\mathcal{G}}
\newcommand{\V}{V}

\newcommand{\pabgxy}{p(a,b|x,y)}

\newcommand{\tr}[1]{\mathrm{Tr}\left[#1\right]}
\newcommand{\ptr}[2]{\mathrm{Tr}_{#1}\left[#2\right]}
\newcommand{\id}{\mathcal{I}}
\newcommand{\identity}{\mathbb{1}}
\renewcommand{\H}{\mathcal{H}}
\newcommand{\ot}{\otimes}
\renewcommand{\S}{\mathcal{S}}
\newcommand{\ketbra}[1]{\ket{#1}\!\!\bra{#1}}
\newcommand{\pg}{p_{\mathrm{guess}}}

\newcommand{\hmin}{H_{\min}}

\newcommand{\bin}[2]{\mathrm{Bin}\left(#1,#2\right)}

\newcommand{\D}{\mathcal{D}}
\newcommand{\DD}{\mathfrak{D}}
\newcommand{\iid}{i$.$i$.$d$.$ }
\newcommand{\x}{\tilde{x}}
\newcommand{\y}{\tilde{y}}
\newcommand{\out}{ab}
\newcommand{\inp}{xy}
\newcommand{\outc}{ab}
\newcommand{\inpc}{x,y}
\newcommand{\Out}{AB}
\newcommand{\Outc}{A,B}

\newcommand{\eps}{\varepsilon}
\newcommand{\epeat}{\epsilon_{\mathrm{EAT}}}
\newcommand{\epsmooth}{\epsilon_{s}}

\renewcommand{\AA}{\bm{\mathrm{A}}}
\newcommand{\BB}{\bm{\mathrm{B}}}
\newcommand{\CC}{\bm{\mathrm{C}}}

\newcommand{\XX}{\bm{\mathrm{X}}}
\newcommand{\YY}{\bm{\mathrm{Y}}}
\newcommand{\II}{\bm{\mathrm{I}}}
\newcommand{\EE}{\mathrm{E}}

\newcommand{\TT}{\bm{\mathrm{T}}}
\newcommand{\N}{\mathcal{N}}
\renewcommand{\d}{\bm{\delta}}
\newcommand{\dpm}{\d_{\sgn}}
\newcommand{\ext}{R_{\mathrm{ext}}}
\newcommand{\zraw}{\AA\BB}
\newcommand{\zout}{\bm{\mathrm{Z}}}
\newcommand{\zseed}{\bm{\mathrm{D}}}
\newcommand{\epext}{\epsilon_{\mathrm{ext}}}
\newcommand{\eploss}{\ell_{\mathrm{ext}}}
\newcommand{\pabort}{\mathrm{P}_{\mathrm{abort}}}
\newcommand{\epcomp}{\varepsilon_{\mathrm{comp}}}
\newcommand{\epsound}{\varepsilon_{\mathrm{sound}}}

\renewcommand{\l}{\bm{\lambda}}
\newcommand{\w}{\bm{\omega}}
\newcommand{\p}{\bm{p}}
\renewcommand{\v}{\bm{\nu}}

\newcommand{\lv}{\l_{\v}}

\newcommand{\lw}{\l_{\w}}

\newcommand{\pr}[1]{\mathrm{Pr}\left[#1\right]}
\newcommand{\E}[1]{\mathbb{E}\left[#1\right]}

\newcommand{\ind}[1]{\multido{}{#1}{\pcind}}
\newcommand{\nln}[1]{{\scriptsize #1:}}
\newcommand{\tbf}[1]{\textbf{#1}}

\newcommand{\sgn}{\mathrm{sgn}}
\newcommand{\gen}{\mathrm{gen}}
\newcommand{\test}{\mathrm{test}}

\newcommand{\chsh}{\mathrm{CHSH}}
\newcommand{\cross}{\mathrm{align}}
\newcommand{\FB}{\mathrm{EB}}
\newcommand{\al}{\braket{AB}}

\newcommand{\supp}{\mathrm{supp}}

\newcommand{\kmax}{k_{\max}}
\newcommand{\epsamp}{\epsilon_{\mathrm{RIA}}}
\newcommand{\epdist}{\epsilon_{\mathrm{dist}}}
\newcommand{\nseed}{N_{\max}}

\newcommand{\ceil}[1]{\left\lceil #1 \right\rceil}
\newcommand{\floor}[1]{\left\lfloor #1 \right\rfloor}
\newcommand{\e}{\bm{\mathrm{e}}}

\newcommand{\rhotheta}{\rho_{\theta}}

\newcommand{\wch}{\omega_{\chsh}}
\newcommand{\wcr}{\omega_{\cross}}

\newcommand{\expo}{\mathrm{e}}
\newcommand{\protoName}{QRE}

\newcommand{\Qk}{\Q^{(k)}}
\newcommand{\pk}{p^{(k)}_{\text{\scriptsize \selectfont \hspace{-0.5pt}guess}}}
\newcommand{\dk}{d^{(k)}_{\text{\scriptsize \selectfont \hspace{-0pt}guess}}}

\newcommand{\errV}{\epsilon_V}
\newcommand{\errK}{\epsilon_K}
\newcommand{\errW}{\epsilon_\Omega}

\newcommand{\Max}[1]{\mathrm{Max}[#1]}
\newcommand{\Min}[1]{\mathrm{Min}[#1]}
\newcommand{\Var}[1]{\mathrm{Var}[#1]}
\newcommand{\respecting}{\Gamma} %Protocol Respecting set
\newcommand{\frestricted}{\left.f\right|_{\respecting}}

\newcommand{\lmin}{\lambda_{\min}}
\newcommand{\lmax}{\lambda_{\max}}
\newcommand{\Av}{A_{\v}}
\newcommand{\Bv}{B_{\v}}

\newcommand{\QG}{\Q_{\G}}

\newcommand{\Ci}{C_i}

\newcommand{\errDFR}{\epsilon_{\text{DFR16}}}
\newcommand{\smax}{s_{\max}}
\newcommand{\sbar}{\bar{s}}

\newcommand{\Fmin}{\mathcal{F}_{\min}}
\newcommand{\DF}{\mathrm{DF}18}
\newcommand{\DFR}{\mathrm{DFR}16}

\newcommand{\PABGXY}{\P_{\A\B|\X\Y}}
\newcommand{\QABGXY}{\Q_{\A\B|\X\Y}}

\newcommand{\freqc}{F_{\CC}}
\newcommand{\freqcbm}{\bm{F}_{\CC}}
\newcommand{\freqexp}{\w}

\newcommand{\noabort}{\Omega}
\newcommand{\cnorm}{c_{\text{norm}}}

\begin{document}

\author{Peter J. Brown\footnote{peter.brown@york.ac.uk}}
\author{Sammy Ragy\footnote{sammy.ragy@york.ac.uk}}
\author{Roger Colbeck\footnote{roger.colbeck@york.ac.uk}}
\affil{Department of Mathematics, University of York, Heslington, York YO10 5DD, United Kingdom}

\renewcommand\Affilfont{\itshape\small}
\title{A framework for quantum-secure device-independent
  randomness expansion}
\maketitle

\begin{abstract}
A device-independent randomness expansion protocol aims to take an initial random seed and generate a longer one without relying on details of how the devices operate for security. A large amount of work to date has focussed on a particular protocol based on spot-checking devices using the CHSH inequality. Here we show how to derive randomness expansion rates for a wide range of protocols, with security against a quantum adversary. Our technique uses semidefinite programming and a recent improvement of the entropy accumulation theorem. To support the work and facilitate its use, we provide code that can generate lower bounds on the amount of randomness that can be output based on the measured quantities in the protocol. As an application, we give a protocol that robustly generates up to two bits of randomness per entangled qubit pair, which is twice that established in existing analyses of the spot-checking CHSH protocol in the low noise regime. 
\end{abstract}

\section{Introduction}
Random numbers are an essential resource in the information processing
era, finding applications in gaming, simulations and
cryptography. Cryptographic protocols are frequently built upon an
assumption of access to a private random seed. Using poor-quality
randomness can be fatal to the security of the protocol (see,
e.g.,~\cite{weakkeys12}). Thus, in order to adhere to these standard
protocol assumptions, it is imperative that we are able to certify the
generation of private random numbers.

The intrinsic randomness of quantum theory provides a
natural mechanism with which we can generate random numbers: a simple
source of perfectly random bits could be a device that prepares a
$\sigma_x$ eigenstate and then measures $\sigma_z$. However, the use
of such a source comes with a significant caveat: the internal
mechanisms of the preparation and measurement devices must be
well-characterized and kept stable throughout their use. Any mismatch
between the characterization and how the device operates in practice
may be an exploitable weakness in the hands of a smart enough
adversary; such mismatches have been used to compromise commercially
available quantum key distribution (QKD)
systems (see e.g.,~\cite{Mak17-qkd_attack}).

While weaknesses caused by mismatches may be mitigated by increasingly
detailed descriptions of the quantum devices, generating such
descriptions rapidly becomes unwieldy and remaining vulnerabilities
can be difficult to detect. This is reminiscent of the situation in
modern software engineering where security flaws are frequently
discovered and patched. Fixing hardware vulnerabilities, such as those
exploited in the aforementioned QKD attacks, can be more difficult
logistically and economically.

Fortunately, quantum theory provides a means to address this
problem. Going back to~\cite{MayersYao} and using an important insight
of~\cite{Ekert}, device-independent quantum cryptography establishes security independently of the devices involved within a protocol, relying only on the validity of quantum theory and the imposition of certain no-signalling constraints between devices. Security is subsequently verified through the observation of non-local output statistics, which in turn act as witnesses to the inner workings of the devices. Limiting the number of initial assumptions greatly reduces the threat of side-channel attacks. 

In this work we focus on the task of randomness expansion: a procedure wherein one attempts to transform a short private seed into a much larger (still private) source of uniform random bits. Randomness expansion in a device-independent setting was proposed
in~\cite{ColbeckThesis,CK2} with further development and experimental
testing following shortly after~\cite{PAMBMMOHLMM}. Subsequent work
provided security proofs against classical
adversaries~\cite{PM,FGS}. Security against quantum adversaries---who
may share entanglement with the internal state of the device---came
later~\cite{MS2,MS1,VV}, progressively increasing in noise-tolerance
and generality, with the recently introduced entropy accumulation
theorem (EAT)~\cite{DFR,DF}, on which our work is based, providing
asymptotically optimal rates~\cite{ARV,arnon2018practical}. A new
proof technique that is also asymptotically optimal has recently
appeared~\cite{KZF}.

In~\cite{ARV} the EAT was
applied to the task of randomness expansion and a general entropy
accumulation procedure was detailed. The security of the resulting
randomness expansion protocol relies on the construction of a
randomness bounding function (known as a min-tradeoff function) that
characterizes the average entropy gain during the protocol. Unfortunately,
the analysis in~\cite{ARV} applies only to protocols based on the CHSH
inequality, and relies on some analytic steps that do not directly
generalize to arbitrary protocols\footnote{In particular,
	simplifications that arise due to the two party, two input, two
	output scenario being reducible to qubits.}. However, as was also noted in \cite{ARV}, one could look to use the device-independent guessing probability (DIGP)~\cite{NPS14,
	BSS14,KRS} in
conjunction with the semidefinite hierarchy~\cite{NPA,NPA2} to obtain
computational constructions of the required min-tradeoff functions. 

Here we detail a precise method for combining these semidefinite
programming tools with the EAT to construct min-tradeoff functions. We
then apply this construction to the task of randomness expansion to
prove security of protocols based upon arbitrary nonlocal games. This
includes protocols with arbitrary (but finite) numbers of
inputs-outputs, as well as protocols based upon multiple
Bell-inequalities~\cite{NBS16}. It is worth noting that this
construction could also be readily extended to multipartite scenarios
although we do not discuss these in this work.  Moreover, as this
computational method takes the form of a semidefinite program these
constructions are both computationally efficient and reliable,
although at the cost of producing potentially suboptimal bounds. To
accompany this work, we provide a code package (available
at~\cite{dirng-github}) for the construction and analysis of these
randomness expansion protocols.

%The majority of device-independent security proofs rely on a single
%Bell-expression for testing non-locality\footnote{An exception to this
%  is~\cite{NBS16} where security was established against a classical
%  adversary.}. The technique we introduce here goes beyond this,
%allowing us to prove security with an arbitrary number of devices,
%subject to an arbitrary number of Bell-expressions on arbitrary
%(finite) alphabets. The flexibility of this construction grants a user
%the freedom to tune the randomness expansion protocol depending on the
%devices they have to hand. Furthermore, the computational nature of
%the security proof allows for easily calculable expansion rates,
%facilitating the exploration of alternate protocols.

In more detail, we give a template protocol, Protocol~\protoName, from
which a user can develop their randomness expansion protocol. Given
certain parameters chosen by a user, e.g., time constraints, choice of
non-locality tests and security tolerances, the projected randomness
expansion rates to be calculated. If these rates are unsatisfactory,
then modifications to the protocol's design can be made and the
rates recalculated. As the computations can be done with a computationally efficient procedure, the user can optimize
their protocol parameters to best fit their experimental
setup. Once a choice of experimental design has been made, the
resulting randomness expansion procedure can be performed. Subject to
the protocol not aborting, this gives a certifiably private random
bit-string.

We apply our technique to several example protocols. In particular, we
look at randomness expansion using the complete empirical distribution
as well as a simple extension of the CHSH protocol, showing
noise-tolerant rates of up to two bits per entangled qubit pair,
secure against quantum adversaries. Although means of generating two
bits of randomness per entangled qubit pair have been considered
before~\cite{MP13} to the best of our knowledge our work is the first
to present a full protocol and prove that this rate can be robustly
achieved taking into account finite statistics.  The nonlocal game we
use for this is related to that in~\cite{MP13}. We also compare the
achievable rates for these protocols to the protocol presented
in~\cite{ARV} which is based upon a direct von Neumann entropy bound.
Our comparison demonstrates that some of the protocols from the
framework are capable of achieving higher rates than the protocol
of~\cite{ARV}, in both the low and high noise regimes. Improved rates
in the high noise regime are of particular importance when considering
current experimental implementations, because of the difficulty of
significantly violating the CHSH inequality while closing the
detection loophole~\cite{Giustina,Shalm,Hensen}. Additionally, we include
in the appendices a full non-asymptotic account of input randomness
necessary for running the protocols.

The paper is structured as follows: in \sec\ref{sec:prelim} we
introduce the material relevant for our construction. In
\sec\ref{sec:adaptive-framework} we detail the various components of
our framework and present the template protocol with full security statements and proofs. We provide examples of several randomness expansion protocols built within our framework in \sec\ref{sec:examples} before concluding with some open problems in \sec\ref{sec:conclusion}.

\section{Preliminaries}\label{sec:prelim}
\subsection{General notation}\label{sec:notation}
%For notational ease, we consider only the bipartite case in this work;
%the generalization to more parties is straightforward.
%% already said earlier

Throughout this work, the calligraphic symbols $\A$, $\B$, $\X$ and
$\Y$ denote finite alphabets and we use the notational shorthand
$\A\B$ to denote the Cartesian product alphabet $\A\times\B$. We refer
to a \emph{behaviour} (or \emph{strategy}) on these alphabets as some
conditional probability distribution, $(p(a,b|x,y))_{ab|xy}$ with
$abxy \in \ABXY$, which we view as a vector $\p\in\R^{|\ABXY|}$. That is, by denoting the set of canonical bases vectors of $\R^{|\ABXY|}$ by $\{\e_{ab|xy}\}_{abxy}$, we write $\p = \sum_{abxy} p(a,b|x,y) \e_{ab|xy}$. We make the distinction between the vector and its elements through the use of boldface, i.e., $p(a,b|x,y) = \bm p \cdot \e_{ab|xy}$. Throughout this work we assume that all
conditional distributions obey the no-signalling constraints that
$\sum_{a\in \A} p(a,b|x,y)$ is independent of $x$ and hence can be written $p(b|y)$ and similarly
$\sum_{b\in\B} p(a,b|x,y) = p(a|x)$. We denote the set of all no-signalling
behaviours by $\P_{\A\B |\X\Y} \subset \R^{|\ABXY|}$. Given an
alphabet $\C$ we denote the set of all distributions over $\C$ by
$\P_\C$, and given a
sequence $\CC = (c_i)_{i=1}^n$, with $c_i \in \C$ for each
$i=1,\dots,n$, we denote the frequency distribution induced by $\CC$ by
\begin{equation}
\freqc(x) = \frac{\sum_{i=1}^n \delta_{x c_i}}{n},
\end{equation}
where $\delta_{ab}$ is the Kronecker delta on the set $\C$.

We use the symbol $\H$ to denote a Hilbert space, subscripting with
system labels when helpful. For a system $E$, we denote the set of
positive semidefinite operators with unit trace acting on $\H_E$ by
$\S(E)$ and its subnormalized extension (i.e., the set that arises
when the trace is restricted to be in the interval $[0,1]$) by
$\tilde\S(E)$ (we extend the use of tildes to other sets to denote
their subnormalized extensions). We refer to a state
$\rho_{XE}\in\S(XE)$ as a \emph{classical-quantum state} (cq-state) on
the joint system $XE$ if it can be written in the form
$\rho_{XE}=\sum_x p(x)\ketbra{x}\otimes\rho_E^x$ where $\{\ket{x}\}_x$
is a set of orthonormal vectors in $\H_X$. Letting $\Omega\subseteq\X$
be an event on the alphabet $\X$, we define the \emph{conditional
  state} (conditioned on the event $\Omega$) by
\begin{equation}
\rho_{XE|{\Omega}}=\frac{1}{\pr{\Omega}}\sum_{x\in\Omega}p(x)\ketbra{x}\ot\rho_E^x.
\end{equation} We denote the identity operator of a system $E$ by
$\identity_E$. We write the natural logarithm as $\ln(\cdot)$ and the
logarithm base $2$ as $\log(\cdot)$. The function
$\sgn:\R\rightarrow\{-1,0,1\}$ is the sign function, mapping all
positive numbers to $1$, negative numbers to $-1$ and $0$ to $0$.

We say that a behaviour $\p\in\PABGXY$ is \emph{quantum} if its
elements can be written in the form
$p(a,b|x,y)=\tr{\rho_{AB}(N_{a|x}\ot M_{b|y})}$ where
$\rho_{AB}\in\S(\H_A\ot\H_B)$ and $\{\{N_{a|x}\}_{a\in\A}\}_{x\in\X}$,
$\{\{M_{b|y}\}_{b\in\B}\}_{y\in\Y}$ are sets of POVMs; we denote the
set of all quantum behaviours by $\Q$. Additionally, we use $\Qsub$ to
denote the subnormalized extension of this set.

Note that randomness expansion is a single-party protocol; there is
one user who wishes to expand an initial private random
string. However, that user may work with a bipartite setup in which
they use two devices that are prevented from signalling to one
another; in such a case we sometimes refer to Alice and Bob as the
users of each device.  Note though that, unlike in QKD, Alice and Bob
are agents of the same party and are within the same laboratory.
There may also be a dishonest party, Eve, trying to gain information
about the random outputs.

\subsection{Entropies and SDPs}\label{sec:entropies-and-sdps}
The von Neumann entropy of $\rho\in\S(A)$ is 
\begin{equation}
H(A)_{\rho}:=-\tr{\rho\log(\rho)}.
\end{equation}
For a bipartite state $\rho_{AE}\in\S(AE)$, we use the notation
$\rho_{E}$ for $\ptr{A}{\rho_{AE}}$ and define the conditional von Neumann entropy of system $A$ given system $E$ when the joint system is in state $\rho_{AE}$ by 
\begin{equation}
H(A|E)_{\rho}:=H(AE)_{\rho}-H(E)_{\rho}\,.
\end{equation}
In addition, for a tripartite system $\rho_{ABE}\in\S(ABE)$, the
conditional mutual information between $A$ and $B$ given $E$ is
defined by
$$I(A:B|E)_{\rho}=H(A|BE)_{\rho}-H(A|E)_{\rho}\,.$$
We drop the state subscript whenever the state is clear from the
context.

In this work it is useful to consider the conditional
min-entropy~\cite{Renner} in its operational
formulation~\cite{KRS}. Given a cq-state $\rho_{XE}=\sum_x p(x)\ketbra{x}\otimes \rho_E^x$, the maximum probability with which an agent holding system $E$ can guess the outcome of a measurement on $X$ is
\begin{equation}\label{eq:g-prob}
\pg(X|E):=\max_{\{M_x\}_x}\sum_x p(x)\tr{M_x\rho_E^x},
\end{equation}
where the maximum is taken over all POVMs $\{M_x\}_x$ on system $E$. Using this we can define the min-entropy of a classical system given quantum side information as 
\begin{equation}\label{eq:hmin}
\hmin(X|E):=-\log\left(\pg(X|E)\right). 
\end{equation}
The final entropic quantity we consider is the $\epsilon$\emph{-smooth min-entropy}~\cite{RenWolf04-smooth_entropies}. Given some $\epsilon\geq 0$ and $\rho_{XE}\in\S(XE)$, the $\epsilon$-smooth min-entropy $H_{\min}^\epsilon$ is defined as the supremum of the min-entropy over all states $\epsilon$-close to $\rho_{XE}$,
\begin{equation} \label{eq:smooth-entropy}
\hmin^\epsilon(X|E)_\rho:=\sup_{\rho'\in B_{\epsilon}(\rho)}\hmin(X|E)_{\rho'},
\end{equation} 
where $B_{\epsilon}(\rho)$ is the $\epsilon$-ball centred at $\rho$ defined with respect to the purified trace distance~\cite{TCR2}. For a thorough overview of smooth entropies and their properties we refer the reader to~\cite{Tom15-smooth_entropies_book}.

In the device-independent scenario we do not know the quantum states
or measurements being performed. Instead, our entire knowledge about these must be inferred
from the observed input-output behaviour of the devices
used. In particular, observing correlations that violate a Bell
inequality provides a coarse-grained characterization of the
underlying system. In a device-independent protocol, the idea is to
use only this to infer bounds on particular system quantities, e.g.,
the randomness present in the outputs. As formulated above, the
guessing probability~\eqref{eq:g-prob} is not a device-independent
quantity because its computation requires knowing $\rho^x_E$. However,
the guessing probability can be reformulated in a device-independent
way~\cite{HR,NPS14,BSS14,NBS16} as we now explain.

Consider a tripartite system $\rho_{ABE}$ shared between two devices
in the user's lab and Eve.  Because we are assuming an adversary
limited by quantum theory, we can suppose that, upon receiving some inputs $(x,y) \in \X\Y$, the devices
work by performing measurements $\{E_{a|x}\}_a$ and $\{F_{b|y}\}_b$
respectively, which give rise to some probability distribution $\p \in \Q_{\A\B|xy}$,
and overall state
$$\sigma^{x,y}_{ABE}=\sum_{ab}\ketbra{a}\ot\ketbra{b}\ot\tilde{\rho}_E^{abxy}\,,$$
where
$\ptr{AB}{(E_{a|x} \ot F_{b|y}\ot\identity_E)\rho_{ABE}}=\tilde{\rho}_E^{abxy}$,
and $p(a,b|x,y)=\tr{\tilde{\rho}_E^{abxy}}$.
Note that the user of the protocol is not aware of what the devices
are doing.

Consider the best strategy for Eve to guess the value of $AB$ using her system
$E$.  She can perform a measurement on her system to try to
distinguish $\{\rho_E^{\out\inp}\}_{\out}$ (occurring with probability
$p(a,b|x,y)$).
Denoting Eve's POVM $\{M_c\}_c$
with outcomes in one-to-one correspondence with the values $AB$ can take
(say $c_{\outc}$
being the value corresponding to a best guess of
$AB = (a,b)$)\footnote{Without
  loss of generality we can assume Eve's measurement has as many
  outcomes as what she is trying to guess.}, then given some values of
$a,b,x$ and $y$,
Eve's outcomes are distributed as
$p(c_{a'b'}|a,b,x,y)=\tr{M_{c_{a'b'}}\rho_E^{\out\inp}}$,
and her probability of guessing
correctly is
$p(c_{ab}|a,b,x,y)=\tr{M_{c_{ab}}\rho_E^{\out\inp}}$.
Hence, the overall probability of guessing $AB$
correctly given $E$
and $XY=(\inpc)$
for the quantum realisation of the statistics, $q=\{\rho_{ABE},\{E_{a|x}\},\{F_{b|y}\}\}$,
is
\begin{align*}
\pg(\Out|\inpc,E,q) &= \sup_{\{M_c\}_c}\sum_{\out} \tr{(E_{a|x}\ot F_{b|y}\ot M_{c_{ab}} )\rho_{ABE}} \\
&= \sup_{\{M_c\}_c} \sum_{\out} p(a,b,c_{ab}|x,y,q)\\
&=\sup_{\{M_c\}_c}\sum_{\out}p(c_{\outc}|a,b,x,y,q)p(a,b|x,y,q)\,.
\end{align*}
Note that the guessing probability depends on the inputs $\inpc$.  In the protocols
we consider later, there will only be one pair of inputs for which Eve is
interested in guessing the outputs.  We denote these inputs by $\x$ and $\y$.

In the device-independent scenario, Eve can also optimize over all
quantum states and measurements that could be used by the devices.
However, she wants to do so while restricting the devices to obey
certain relations which depend on the protocol (for example, the CHSH
violation that could be observed by the user).  For the moment,
without specifying these relations precisely, call the set of quantum
states and measurements obeying these relations $\cR$. Hence, we seek
$$\pg(\Out|\x,\y,E)=\sup_{q\in\cR,\{M_c\}_c}\sum_{\out}p(a,b|\x,\y,q)p(c_{\outc}|a,b,\x,\y,q)\,.$$

Because Eve's measurement commutes with those of the devices, due to
no signalling we can use Bayes' rule to rewrite the optimization
as\footnote{This rewriting makes sense provided no information leaks
  to Eve during the protocol, which is reasonable for randomness expansion
  since it takes place in one secure lab.}
\begin{eqnarray*}
\sup_{q\in\cR,\{M_c\}_c}\sum_{\out}p(c_{\outc}|\x,\y,q)p(a,b|c_{ab},\x,\y,q)\,.
\end{eqnarray*}
With this rewriting it is evident that we can think about Eve's
strategy as follows: Eve randomly chooses a value of $C$ and then
prepares the device according to that choice, i.e., trying to bias
$\Outc$ towards the values $a, b$ corresponding to the chosen $c$.

We can hence write
$$\pg(\Out|\x,\y,E)=\sup_{\{\p_c\}_c}\sum_{\outc}\pr{C=c_{ab}}p_{c_{ab}}(a,b|\x,\y,q)\,,$$
where $\sum_cp(c)\p_{c}$ satisfies some relations (equivalent to
the restriction to the set $\cR$) and $\p_{c}\in\Q_{\A\B|\X\Y}$ for each $c$.
Provided the relations satisfied are linear, which we will henceforth
assume, they can be expressed as a matrix equation $\bm{W}\p = \w$ and the whole optimization is a conic program (the set of un-normalized
quantum-realisable distributions forms a convex cone). By writing
$\pr{C=c} \p_{c}$ as the subnormalized distribution
$\tilde{\p}_{c}$ the problem can be expressed as
\begin{equation}\label{prog:digp}
  \begin{aligned}
    \sup_{\{\tilde{\p}_{c}\}_c} &&&\sum_{ab}\tilde{p}_{c_{ab}}(a,b|\x,\y)\\
\text{subj.\
  to}&&&{\sum_c\bm{W}\tilde \p_{c}}=\w\\
&&&\tilde{\p}_{c}\in\Qsub_{\A\B|\X\Y}\quad\forall\ c\,.
\end{aligned}
\end{equation}
Note that the normalisation condition, $\sum_{abc}\tilde{p}_c(a,b|\x,\y)=1$, is assumed to be contained within (or a consequence of) the conditions imposed by $\bm{W}$. For the particular sets of conditions that we impose later, normalization always follows.

Optimizing over the set of quantum correlations is a difficult
problem, in part because the dimension of the quantum system achieving
the optimum could be arbitrarily large.  Because of this, we consider
a computationally tractable relaxation of the problem, by instead
optimizing over distributions within some level of the
semidefinite hierarchy~\cite{NPA,NPA2}.  We denote the $k^{\text{th}}$ level
by $\Qsub^{(k)}$. This relaxation of the problem
takes the form of a semidefinite program that can be solved in an
efficient manner, at the expense of possibly not obtaining the same
optimum value. The corresponding relaxed program is
\begin{equation}\label{prog:relaxed-primal}
\begin{aligned}
\pk(\w):=&\sup_{\{\tilde{\p}_{c}\}_c} &&\sum_{ab}\tilde{p}_{ c_{ab}}(a,b|\x,\y)\\
&\text{subj.\
  to}&&{\sum_c\bm{W}\tilde \p_{c}}=\w\\
&&&\tilde{\p}_{c}\in\Qsub^{(k)}\quad\forall\ c\,.
\end{aligned}
\end{equation}

This program has a dual. In Appendix~\ref{app:cones} we show that
there is an alternative program with the same properties\footnote{In
  particular, the weak duality statement holds.} as the standard dual.
To specify this, we define the set $\cV^{(k)}$ of \emph{valid constraint
  vectors at level $k$} by the set of vectors $\v$ for which there exists
${\p}\in\Q^{(k)}$ such that $\bm{W}{\p}=\v$.

The alternative dual then takes the form
\begin{equation}\label{prog:relaxed-dual}
\begin{aligned}
\dk(\w):=&\inf_{\l} &&\l\cdot\w\\
&\text{subj.\
  to}&
&\pk(\v)\leq \l\cdot\v,\ \ \forall\ \v\in\cV^{(k)},
\end{aligned}
\end{equation}
with $\l\in\R^{\|\w\|_0}$. Since the NPA hierarchy forms a sequence of outer
approximations to the set of quantum correlations,
$\Q_1\supseteq\Q_2\supseteq\dots\supseteq\Q$, the relaxed
guessing probability provides an upper bound on the true guessing
probability, i.e., $\pg(\w)\leq \pk(\w)$. Combined
with~\eqref{eq:hmin}, one can use the relaxed programs to compute
valid device-independent lower bounds on $\hmin$.

Programs~\eqref{prog:relaxed-primal} and~\eqref{prog:relaxed-dual} are
parameterized by a vector $\w$. We denote a feasible point of the dual
program parameterized by $\w$ by $\lw$. Note that for our later
analysis we only need $\lw$ to be a feasible point of the dual
program, we do not require it to be optimal.\footnote{An optimal choice of $\l$ for \eqref{prog:relaxed-dual} may not even exist.}

\subsection{Devices and nonlocal games}\label{sec:devices-and-games}
Device-independent protocols involve a series
of interactions with some \emph{untrusted devices}. A \emph{device}
$\D$ refers to some physical system that receives classical inputs
and produces classical outputs. Furthermore, we say that $\D$ is
\emph{untrusted} if the mechanism by which $\D$ produces the outputs
from the inputs need not be characterized. During the protocol, the user interacts with their untrusted devices within the following scenario:\footnote{One does not have to recreate this scenario exactly in order to perform the protocol. Instead, the given scenario establishes one situation in which the protocol remains secure (see \defn\ref{def:security} for a precise definition of security).}
\begin{enumerate}
\item The protocol is performed within a secure lab from which
  information can be prevented from leaking.
\item This lab can be partitioned into disconnected sites
  (one controlled by Alice and one by Bob).
\item The user can send information freely between these sites without
  being overheard, while at the same time, they can prevent
  unwanted information transfer between the sites.\footnote{In this
    work we need to ensure that the user's devices are unable
    to communicate at certain points of the protocol (when Bell tests
    are being done), but not at others (e.g., when entanglement is
    being distributed).  However, they should never be allowed to
    send any information outside the lab after the protocol begins.}
\item The user has two devices to which they can provide inputs (taken
  from alphabets $\X$ and $\Y$) and receive outputs (from alphabets
  $\A$ and $\B$).
\item These devices operate according to quantum theory, i.e.,
  $\p_{AB|XY}\in\QABGXY$.  Any eavesdropper is also limited by
  quantum theory\footnote{In parts of this paper we allow the
    eavesdropper limited additional power---the bounds will then still
    apply if the eavesdropper is limited by quantum theory.}.  We use
  $\DD_{ABE}$ to denote the collection of devices (including any held
  by an eavesdropper) and refer to this as an \emph{untrusted device
    network}.
\item The user has an initial source of private random numbers and a
  trusted device for classical information processing.
\end{enumerate}

One of the key advantages of a device-independent protocol is that
because no assumptions are made on the inner workings of the devices
used, the protocol checks that the devices are working sufficiently
well on-the-fly. The protocols hence remain impervious to many
side-channel attacks, malfunctioning devices or prior tampering. The idea behind their security is that by testing that
the devices exhibit `nonlocal' correlations, their internal workings
are sufficiently restricted to enable the task at hand.

In this work, we formulate the testing of the devices through
\emph{nonlocal} games.  A nonlocal game is initiated by a referee who
sends the two players their own question chosen according to some
distribution, $\mu$. The players then respond with their answers
chosen from $\A$ and $\B$ respectively. Using the predefined scoring
rule $V$, the referee then announces whether or not they won the
game. The game is referred to as \emph{nonlocal} because prior to
receiving their questions, the players are separated and unable to
communicate until they have given their answers. The question sets,
answer sets, distribution $\mu$ and the scoring rule $V$ are all
public knowledge. Moreover, the players are allowed to confer prior to
the start of the game.
	 \begin{definition}\label{def:nonlocalgame}
	 	Let $\A, \B, \X, \Y$ and $\mathcal{V}$ be finite sets. A (two-player) \textit{nonlocal game} $\G=(\mu,V)$  (on $\ABXY$) consists of a set of question
	 	pairs $(x,y) \in \X\Y$ chosen according to some probability
	 	distribution $\mu:\X\Y \rightarrow [0,1]$, a set of answer pairs
	 	$(a,b) \in \A\B$ and a scoring function
	 	$V:\A\B\X\Y\rightarrow \mathcal{V}$. A \emph{strategy} for $\G$ is a conditional distribution
	 	$\p\in \PABGXY$ defined on the question and answer sets.
	 \end{definition}
 \begin{remark}
 	We will abuse notation and use the symbol $\G$ to refer to both the nonlocal game and the set of possible scores. I.e., we may refer to the players receiving a score $c \in \G$. Furthermore, we denote the number of different scores by $|\G|$.
 \end{remark}
 
  If the players play $\G$ using the strategy $\p$, then this induces a frequency distribution $\w_{\G}$ over the set of possible scores. That is, 
  \begin{equation}\label{eq:expected-freq-dist}
  \omega_{\G}(c) = \sum_{abxy} \mu(x,y)p(a,b|x,y)\,\delta_{V(a,b,x,y), c}
  \end{equation}
  for each $c \in \G$. The expected frequency distribution, $\w_{\G}$, will be the figure of merit by which we evaluate the performance of our untrusted devices. We denote the set of possible frequency distributions achievable by the agents whilst playing according to quantum strategies by $\Q_\G$.

\begin{example}[Extended CHSH game $(\G_{\chsh})$]\label{ex:chsh}
	The \emph{extended CHSH game} has appeared already in the device-independent literature in the context of QKD (see, e.g.,~\cite{ABGMPS}). It extends the standard CHSH game to include a correlation check between one of Alice's CHSH inputs and an additional input from Bob. It is defined by the question-answer sets $\X=\{0,1\}$, $\Y=\{0,1,2\}$ and $\A=\B=\{0,1\}$, the scoring set $\mathcal{V} = \{c_{\chsh}, c_{\cross},0\}$ and the scoring rule
	\begin{align}
	V_{\chsh}(a,b,x,y) &:=
	\begin{cases}
	c_{\chsh} \qquad \qquad &\text{if } x\cdot y = a\oplus b \text{ and } y\neq2 \\
	c_{\cross} \qquad \qquad &\text{if } (x,y) = (0,2) \text{ and } a \oplus b = 0 \\ 
	0 \qquad \qquad &\text{otherwise}.
	\end{cases}
	\end{align}
	The input distribution we consider is defined by $\mu_{\chsh}(x,y)=\frac18$ for
	$(x,y) \in\{0,1\}^2$, $\mu_{\chsh}(0,2)=\frac12$ and $\mu_{\chsh}(x,y)=0$
	otherwise. This game is equivalent to choosing to play either the $\chsh$ game or
	the game corresponding to checking the alignment of the inputs $(0,2)$ uniformly at random and then proceeding with the chosen
	game. The frequency distribution then tells us the relative frequencies with which we win each game. The motivation behind $\G_{\chsh}$ can be understood by considering a schematic of an ideal implementation on a bipartite qubit system as given in \fig\ref{fig:EkertDiag}. If we observe the maximum winning probability for the CHSH game, as well as perfect alignment for the inputs $(0,2)$, then the inputs $(\x,\y) = (1,2)$ should produce two perfectly uniform bits. 
	
	\begin{figure}[t!]
		\begin{center}
			\begin{tikzpicture}[scale=1]
			\begin{scope}[shift={(0.5,0)}]%
			%Device Boxes + names
			\draw[thick, rounded corners = 2, fill = lightgray] (-1,1,1) rectangle (2,5.5,1);
			\draw[thick, rounded corners = 2, fill = lightgray] (-1,5.5,1) -- (-1,5.5,0) -- (2,5.5,0) -- (2,5.5,1) -- cycle;
			\draw[thick, rounded corners = 2, fill = lightgray] (2,5.5,1) -- (2,5.5,0) -- (2,1,0) -- (2,1,1) -- cycle;
			
			\draw[thick, rounded corners = 2, fill = lightgray] (9.75,1,1) rectangle (12.75,5.5,1);
			\draw[thick, rounded corners = 2, fill = lightgray] (9.75,5.5,1) -- (9.75,5.5,0) -- (12.75,5.5,0) -- (12.75,5.5,1) -- cycle;
			\draw[thick, rounded corners = 2, fill = lightgray] (12.75,5.5,1) -- (12.75,5.5,0) -- (12.75,1,0) -- (12.75,1,1) -- cycle;
			
			\draw (0.25,5.3) node {\tbf{\textit{Alice's device}}};
			\draw (11,5.3) node {\tbf{\textit{Bob's device}}};
			\end{scope}
			%Circles
			\def\x{0};
			\def\y{4};
			\draw[fill=gray!25] (\x,\y) circle (.5cm);
			\draw[dotted] (\x-0.5,\y) -- (\x+.5,\y);
			\draw[dotted] (\x,\y-.5) -- (\x,\y+.5);
			\draw (\x+1.2,\y) node {$X=0$};
			\draw (\x,\y-.7) node {$\sigma_z$};
			\draw [<->,thick] (\x,\y+.5) -- (\x,\y-.5);

			\def\x{0};
			\def\y{2};
			\draw[fill=gray!25] (\x,\y) circle (.5cm);
			\draw[dotted] (\x-0.5,\y) -- (\x+.5,\y);
			\draw[dotted] (\x,\y-.5) -- (\x,\y+.5);
			\draw (\x+1.2,\y) node {$X=1$};
			\draw (\x,\y-.7) node {$\sigma_x$};
			\draw [<->,thick] (\x+.5,\y) -- (\x-.5,\y);
			
			\def\x{12};
			\def\y{4.5};
			\draw[fill=gray!25] (\x,\y) circle (.5cm);
			\draw[dotted] (\x-0.5,\y) -- (\x+.5,\y);
			\draw[dotted] (\x,\y-.5) -- (\x,\y+.5);
			\draw (\x-1.2,\y) node {$Y=0$};
			\draw (\x,\y-.7) node {$\sigma_{\pi/4}$};
			\draw [<->,thick] (\x+0.35,\y+0.35) -- (\x-0.35,\y-0.35);
			
			\def\x{12};
			\def\y{3};
			\draw[fill=gray!25] (\x,\y) circle (.5cm);
			\draw[dotted] (\x-0.5,\y) -- (\x+.5,\y);
			\draw[dotted] (\x,\y-.5) -- (\x,\y+.5);
			\draw (\x-1.2,\y) node {$Y=1$};
			\draw (\x,\y-.7) node {$\sigma_{-\pi/4}$};
			\draw [<->,thick] (\x-0.35,\y+0.35) -- (\x+0.35,\y-0.35);

			\def\x{12};
			\def\y{1.5};
			\draw[fill=gray!25] (\x,\y) circle (.5cm);
			\draw[dotted] (\x-0.5,\y) -- (\x+.5,\y);
			\draw[dotted] (\x,\y-.5) -- (\x,\y+.5);
			\draw (\x-1.2,\y) node {$Y=2$};
			\draw (\x,\y-.7) node {$\sigma_{z}$};
			\draw [<->,thick] (\x,\y+.5) -- (\x,\y-.5);

			%State
			\draw (6,-0.5) node (state) {$\ket{\psi}_{AB} = \frac{1}{\sqrt{2}} (\ket{00} +\ket{11})$};
			\draw[->, line width=0.4mm, dotted] (state) .. controls (1, -.5) and (0,-.5) ..  (0.5,0.5);
			\draw[->, line width=0.4mm, dotted] (state) .. controls (10, -.5) and (12,-.5) ..  (11.5,0.5);
			
			\node[scale = 0.75] (table) at (6.2,4.5) {%
				\begin{tabular}{c c|c|c|c|c|c|c|}
				\multicolumn{2}{l|}{\parbox[t][4mm]{4mm}{\multirow{2}{*}{\rotatebox[origin=l]{45}{$P_{AB|XY}$\,}}}} & \multicolumn{2}{c|}{$Y=0$} & \multicolumn{2}{c|}{$Y=1$} & \multicolumn{2}{c|}{$Y=2$} \tabularnewline
				\hhline{~~------}
				&& $B=0$ & $B=1$ & $B=0$ & $B=1$ & $B=0$ & $B=1$ \tabularnewline \hline
				\multicolumn{1}{c|}{\bigstrut[1] \multirow{2}{*}{\rotatebox[origin=c]{90}{$X=0$\,\,\,\,\,}}} &\parbox[c][12mm]{1mm}{\rotatebox[origin=b]{90}{$A = 0$}}&
				\cellcolor{green!10}{\large$\frac{2 + \sqrt{2}}{8}$ }& {\large$\frac{2 - \sqrt{2}}{8}$ }& \cellcolor{green!10}{\large$\frac{2 + \sqrt{2}}{8}$}& {\large$\frac{2 - \sqrt{2}}{8}$} & \cellcolor{blue!10}{\large$\frac{1}{2}$ }&{\large $0$}\tabularnewline
				\hhline{~-------}
				\multicolumn{1}{c|}{~} & \parbox[c][12mm]{1mm}{\rotatebox[origin=c]{90}{$A = 1$}}
				&{\large $\frac{2 - \sqrt{2}}{8}$} & {\cellcolor{green!10}}{\large$\frac{2 + \sqrt{2}}{8}$} & {\large$\frac{2 - \sqrt{2}}{8}$} & \cellcolor{green!10}{\large$\frac{2 + \sqrt{2}}{8}$} & {\large$0$} & \cellcolor{blue!10}{\large$\frac{1}{2}$}\tabularnewline
				\hline
				\multicolumn{1}{c|}{\multirow{2}{*}{\rotatebox[origin=c]{90}{$X=1$\,\,\,\,\,}}} &\parbox[c][12mm]{1mm}{\rotatebox[origin=b]{90}{$A = 0$}}
				& \cellcolor{green!10}{\large$\frac{2 + \sqrt{2}}{8}$} & {\large$\frac{2 - \sqrt{2}}{8}$} & {\large$\frac{2 - \sqrt{2}}{8}$ }& \cellcolor{green!10}{\large$\frac{2 + \sqrt{2}}{8}$} & \cellcolor{red!10}{\large$\frac{1}{4}$} & \cellcolor{red!10}{\large$\frac{1}{4}$}\tabularnewline
				\hhline{~-------}
				\multicolumn{1}{c|}{~} & \parbox[c][12mm]{1mm}{\rotatebox[origin=c]{90}{$A = 1$}} 
				& {\large$\frac{2 - \sqrt{2}}{8}$} & \cellcolor{green!10}{\large$\frac{2 + \sqrt{2}}{8}$ }& \cellcolor{green!10}{\large$\frac{2 + \sqrt{2}}{8}$} & {\large$\frac{2 - \sqrt{2}}{8}$} & \cellcolor{red!10}{\large$\frac{1}{4}$ }& \cellcolor{red!10}{\large$\frac{1}{4}$}\tabularnewline
				\hline 
				\end{tabular}
			};
			\node(CHSH) [below left =0.cm and -2.1cm of table,align=left] {CHSH};
			\node(chshbox)[draw, left=.1cm of CHSH,  fill=green!10] {};
			\node(align) [below right =0.cm and -1.31cm of CHSH, align=left] {Alignment};
			\node[draw, left=.12cm of align,  fill=blue!10] {};
			\node(gen) [below right =-0.cm and -2.1cm of table] {Generation};
			\node[draw, left=.12cm of gen,  fill=red!10] {};
			\end{tikzpicture} 
		\end{center}
		
		\caption[Qubit implementation of extended CHSH game]{A measurement schematic for a qubit implementation of $\G_{\chsh}$. Measurements are depicted in the $x$-$z$ plane of the
			Bloch-sphere with
			$\sigma_{\varphi} = \cos(\varphi) \sigma_z + \sin(\varphi) \sigma_x$
			for $\varphi \in (-\pi,\pi]$. Using the maximally entangled state
			$\ket{\psi}_{AB}=(\ket{00}+\ket{11})/\sqrt{2}$ with the measurements
			depicted, one has a frequency distribution of
			$\w_{\G} = \tfrac12~\left(\tfrac12 + \tfrac{\sqrt{2}}{4}, 1, \tfrac12 - \tfrac{\sqrt{2}}{4}\right)$, where the scores are ordered $(c_{\chsh}, c_{\cross}, 0)$. The setup achieves Tsirelson's bound for the CHSH game as well as
			perfect correlations for the $X=0$ and $Y=2$ inputs. In addition,
			self-testing results~\cite{PopescuRohrlich} give a converse result:
			these scores completely characterize the devices up to local
			isometries. This implies that the state used by the devices is
			uncorrelated with an adversary and that the measurement pair
			$(X,Y)=(1,2)$ yields uniformly random results, certifying the
			presence of $2$ bits of private randomness in the outputs.}
		\label{fig:EkertDiag}
	\end{figure}
\end{example}

\subsection{Device-independent randomness expansion protocols and their security}
\label{sec:re-protocols}
A device-independent randomness expansion protocol is a procedure by
which one attempts to use a uniform, trusted seed, $\zseed$, to
produce a longer uniform bit-string, $\zout$, through repeated
interactions with some untrusted devices. We consider so-called
\textit{spot-checking} protocols, which involve two round types:
test-rounds, during which one attempts to produce certificates of
nonlocality, and generation rounds in which a fixed input is given to
the devices and the outputs are recorded. By choosing the rounds
randomly according to a distribution heavily favouring generation
rounds, we are able to reduce the size of the seed whilst sufficiently
constraining the device's behaviour, guaranteeing the presence of
randomness in the outcomes (except with some small probability).

Using the setup described in \sec\ref{sec:devices-and-games}, our template randomness expansion protocol consists of two main steps.
\begin{enumerate}
\item \textbf{Accumulation}: In this phase the user repeatedly
  interacts with the separated devices. Each interaction is randomly
  chosen to be a generation round or a test round in a coordinated way
  using the initial random seed.\footnote{For example, a central source of randomness could be used
  	to choose the round type. This information could then be communicated to each party (in such a way that
  	the devices do not learn this choice).} On generation rounds the devices are
  provided with some fixed inputs $(\x,\y)\in\X\Y$, whereas during test rounds,
  the testing procedure specific to the protocol is followed. After many
  interactions, the recorded outputs are concatenated to give
  $\AA\BB$. Using the statistics collected during test rounds, a
  decision is made about whether to abort or not based on how close
  the observations are to some pre-defined expected device behaviour.
\item \textbf{Extraction}: Subject to the protocol not aborting in the
  accumulation step, a quantum-proof randomness extractor is applied
  to $\zraw$. This maps the partially random $\zraw$ to a shorter
  string $\zout$ that is the output of the protocol.
\end{enumerate}  

We define security of a randomness expansion protocol according to a
composable
definition~\cite{Canetti,Can01,PW,Ben-OrMayers,PR2014}. Using
composable security ensures that the output randomness can be used in
any other application with only an arbitrarily small probability of it
being distinguishable from perfect randomness.  To make this more
precise, consider a hypothetical device that outputs a string $\zout$ that
is uniform and uncorrelated with any information held by an
eavesdropper. In other words, it outputs $\tau_m\ot\rho_E$, where
$\tau_m$ is the maximally mixed state on $m$ qubits.  The ideal
protocol is defined as the protocol that involves first doing the real
protocol, then, in the case of no abort, replacing the output with a
string of the same length taken from the hypothetical device. The
protocol is then said to be $\epsound$-secure ($\epsound$ is called
the \emph{soundness error}) if, when the user either implements the
real or ideal protocol with probability $\frac{1}{2}$, the maximum
probability that a distinguisher can guess which is being implemented
is at most $\frac{1+\epsound}{2}$.  If $\epsound$ is small, then the
real and ideal protocols are virtually indistinguishable.  Defining
the ideal as above ensures that the real and ideal protocols can never
be distinguished by whether or not they abort.  We refer
to~\cite{PR2014} for further discussion of composability in a related
context (that of QKD).

There is a second important parameter of any protocol, its
\emph{completeness error}, which is the probability that an ideal
implementation of the protocol leads to an abort. It is important for
a protocol to have a low completeness error in addition to a low soundness error since a protocol that always aborts is vacuously secure.

\begin{definition}\label{def:security}
  Consider a randomness expansion protocol whose output is denoted by
  $Z$.  Let $\Omega$ be the event that the protocol does not
  abort. The protocol is an $(\epsound,\epcomp)$-randomness expansion
  protocol if it satisfies the following two conditions.
\begin{enumerate}
\item \textbf{Soundness}:
\begin{equation}\label{eq:sound_def}
\frac{1}{2}\mathrm{Pr}[\Omega]\cdot\|\rho_{ZE}-\tau_m\ot\rho_E\|_1\leq \epsound,
\end{equation}
where $E$ is an arbitrary quantum register (which could have been
entangled with the devices used at the start of the protocol), $m$
is the length of the output string $Z$ and $\tau_m$ is the maximally mixed state on a system of dimension $2^m$.
\item \textbf{Completeness}: There exists a set of quantum states and
  measurements such that if the protocol is implemented using those
\begin{equation}
\mathrm{Pr}[\Omega]\geq 1-\epcomp.
\end{equation}
\end{enumerate}
\end{definition}

Although we use a composable security definition to ensure that any
randomness output by the protocol can be used in any scenario,
importantly, this may not apply if the devices used in the protocol
are subsequently reused~\cite{bckone}. Thus, after the protocol the
devices should be kept shielded and not reused until such time as the
randomness generated no longer needs to be kept secure.  How best to
resolve this remains an open problem: the Supplemental Material
of~\cite{bckone} presents candidate protocol modifications (and
modifications to the notion of composability) that may circumvent such
problems.

\subsection{Entropy accumulation}
In order to bound the smooth min entropy
$H_{\min}^{\epsilon}(\AA\BB|\XX\YY E)$ accumulated during the protocol
we employ the EAT~\cite{DFR, DF}.  Roughly speaking, the EAT says that
this min-entropy is proportional to the number of rounds, up to square
root correction factors.  The proportionality constant is the
single-round conditional von Neumann entropy optimized over all states
that can give rise to the observed scores. In its full form, the EAT
is an extension of the asymptotic equipartition property~\cite{TCR} to
a particular non-\iid regime. For the purposes of randomness expansion
we only require a special case of the EAT, which we detail later in
this section.

With the goal of maximising our entropic yield, we use the recently
improved statement of the entropy accumulation
theorem~\cite{DF}.\footnote{We discuss this EAT statement and compare
  it to alternatives in \app\ref{app:EAT}.} For completeness we present
the relevant statements including the accumulation procedure (see also~\cite{arnon2018practical}).

\subsubsection{The entropy accumulation procedure}
\label{sec:accumulation-procedure}

The entropy accumulation procedure prescribes how the user interacts
with their untrusted devices and collects data from them. Before
beginning this procedure a nonlocal game $\G=(\mu,V)$ that is compatible with
the alphabets of the devices is selected. 

A \emph{round} within the entropy accumulation procedure consists of
the user giving an input to each of their devices and recording the
outputs. We use subscripts on random variables to indicate the round
that they are associated with, i.e., $X_iY_i$ are the random variables
describing the joint device inputs for the $i^{\text{th}}$ round. In
addition, boldface will be used to indicate that a random variable
represents the concatenation over all $n$ rounds of the protocol,
$\XX = X_1X_2\dots X_n$.

The accumulation procedure consists of $n \in \mathbb{N}$ separate interactions with the untrusted devices. We refer to a single interaction with the devices as a \emph{round}. A round consists of the user selecting and supplying inputs to the devices, receiving outputs and recording this data. During the $i^{\text{th}}$ round, a random variable $T_i \sim \text{Bernoulli}(\gamma)$ is sampled, for some fixed $\gamma \in (0,1)$, indicating whether the round will be a \emph{generation round} or a \emph{test round}. With probability $1-\gamma$ we have $T_i = 0$ and the round is a generation round. During a generation round, the user supplies the respective devices with the fixed generation inputs $(\x,\y) \in \X\Y$, recording $X_iY_i = (\x,\y)$. They record the outputs received from the devices as $A_i$ and $B_i$ respectively and they record the round's score as $C_i = \perp$. With probability $\gamma$, $T_i = 1$ and the round is a test round. During a test round, inputs $X_iY_i$ are sampled according to the distribution specified by the chosen nonlocal game. The sampled inputs are fed to their respective devices and the outputs received are recorded as $A_iB_i$. The score is computed and recorded as $C_i = V(A_i,B_i,X_i,Y_i)$. The \emph{transcript for round $i$} is the tuple $(A_i,B_i,X_i,Y_i,T_i,C_i)$. After $n$ rounds, the complete transcript for the accumulation procedure is $(\AA,\BB,\XX,\YY,\TT,\CC)$. We denote by $\C$ the set of possible values that $C_i$ can take, i.e. $\C = \G \cup \{\perp\}$.

After the $n$-round transcript has been obtained, the user looks to determine the performance of the untrusted devices and, in turn, certify a lower bound on the total entropy produced, $\hmin^{\epsilon}(\AA\BB|\XX\YY\EE)$. To this end, the user computes the \emph{empirical frequency distribution}
\begin{equation}\label{eq:frequency-distribution-scores}
\freqc(c) = \frac{1}{n}\sum_{i=1}^{n} \delta_{c,C_i}.
\end{equation}
Prior to the accumulation step, the user fixes some frequency
distribution $\freqexp$ corresponding to an expected (or hoped for) behaviour. Should the devices behave in an \iid manner according to $\freqexp$, then concentration bounds tell us that the empirical frequency distribution $\freqcbm$ should be close to $\w$. With this in mind, we define the event that the protocol does not abort by
\begin{equation}\label{eq:success-event}
\noabort = \{\CC \mid \gamma (\bm\freqexp(\G) - \d) < \freqcbm(\G) < \gamma (\bm\freqexp(\G) + \d) \},
\end{equation}
where $\d \in (0,1)^{|\G|}$ is a vector of confidence interval widths satisfying $\bm 0 < \d < \bm\freqexp(\G)$ with all vector inequalities being interpreted as element-wise constraints.

\subsubsection{The entropy accumulation theorem}\label{sec:EAT-defns}

To complete the protocol, uniform randomness needs to be extracted
from the partially random outputs.  Doing so requires the user to
assume a lower bound on the smooth min-entropy (conditioned on any side
information held by an adversary) contained in the devices' outputs
when the protocol does not abort.  If $\epsound$ is very small, then
the assumption must be correct with near certainty.  The EAT
provides a method by which one can compute such a lower
bound. Loosely, the EAT states that if the interaction between the
honest parties occurs in a sequential manner (as described in
\sec\ref{sec:accumulation-procedure}), then with high probability the
uncertainty an adversary has about the outputs is \emph{close} to
their total average uncertainty. As a mathematical statement it is a
particular example of the more general phenomenon of
\emph{concentration of measure} (see~\cite{ledoux2005concentration}
for a general overview). In order to state the EAT precisely, we first
require a few definitions.

\begin{definition}[EAT channels]\label{def:eat-channels}
A set of \emph{EAT channels} $\{\N_i\}_{i=1}^n$ is a collection of trace-preserving and completely-positive maps $\N_i:\S(R_{i-1})\rightarrow\S(A_iB_iX_iY_i\Ci R_i)$ such that for every $i \in [n]$:
\begin{enumerate}
\item $A_i,B_i,X_i,Y_i$ and $\Ci$ are finite dimensional classical
  systems, $R_i$ is an arbitrary quantum system and $\Ci$ is the
  output of a deterministic function of the classical registers
  $A_i,B_i,X_i$ and $Y_i$.
\item For any initial state $\rho_{R_0 E}$, the final state $\rho_{\AA\BB\XX\YY\CC E} = \ptr{R_n}{((\N_n \circ \dots \circ \N_1)\otimes \id_{E}) \rho_{R_0 E}}$ fulfils the Markov chain condition $I(A^{i-1}B^{i-1}\!\! :\! X_iY_i | X^{i-1}Y^{i-1}E ) = 0$ for every $i \in [n]$.
\end{enumerate}
\end{definition}

The EAT channels formalise the notion of
interaction within the protocol. The first condition in
\defn\ref{def:eat-channels} specifies the nature of the information
present within the protocol and, in particular, it restricts the honest
parties' inputs to their devices to be classical in nature. The
arbitrary quantum register $R_i$ represents the quantum
state stored by the separate devices after the $i^{\text{th}}$ round. The second condition
specifies the sequential nature of the protocol.  The channels $\N_i$
describe the joint action of both devices and include the generation
of the randomness needed to choose the settings. The Markov chain
condition implies that the inputs to the devices presented by the
honest parties are conditionally independent of the previous outputs
they have received. Note that by using a trusted private seed to
choose the inputs and supplying the inputs sequentially (as is
done in \sec\ref{sec:accumulation-procedure}), this condition will be
satisfied.\footnote{A public seed can also be used if the Markov chain conditions can be shown to hold. However, one may need to be more careful when dealing with such a scenario. For example, if the entangled states distributed to the devices come from some third-party source, then it should be clear that the state used within the $i^{\text{th}}$ round was prepared independently of the seed used to generate the inputs $X_{i+1}^nY_{i+1}^n$. This could be achieved by choosing inputs $X_{i+1}Y_{i+1}$ using a public seed that was generated after the $i^{\text{th}}$ entangled state has been distributed.} Finally, the adversary is permitted to hold a purification,
$E$ of the initial state shared by the devices, and the state
evolves with the sequential interaction through the application of the
sequence of EAT channels.

As explained above, the EAT allows the elevation of \iid analyses to
the non-\iid setting. To do so requires a so-called \emph{min-tradeoff
  function} which, roughly speaking, gives a lower bound on the
single-round von Neumann entropy produced by any devices that, on
expectation, produce some statistics. In the case of the EAT these
distributions are $\{\freqc\}_{\CC\in\Omega}$, i.e., all frequency
distributions induced from score transcripts $\CC$ that do not lead to
an aborted protocol. The EAT asserts that, under sequential
interaction, an adversary's uncertainty about the outputs of the
non-\iid device will (with high probability) be concentrated within
some interval about the uncertainty produced by these \iid devices. In
particular, a lower bound on this uncertainty can be found by
considering the worst-case \iid device.

\begin{definition}[Min-tradeoff functions]\label{def:fmin}
  Let $\{\N_i\}_{i=1}^n$ be a collection of EAT channels and let
  $\mathcal{C}$ denote the common alphabet of the systems
  $C_1,\dots,C_n$. An affine function $f: \P_\C \rightarrow \R$ is a
  \emph{min-tradeoff function} for the EAT channels
  $\{\N_i\}_{i=1}^n$ if for each $i \in [n]$ it satisfies
\begin{equation}\label{eq:fmin-definition}
f(\p) \leq \inf_{\sigma_{R_{i-1}R'}: \N_i(\sigma)_{\Ci}=\tau_{\p}} H(A_iB_i|X_iY_iR')_{\N_i(\sigma)},
\end{equation} 
where $\tau_{\p}:=\sum_{c\in\mathcal{C}}p(c)\ketbra{c}$, $R'$ is a
register isomorphic to $R_{i-1}$ and the infimum over the empty set is
taken to be $+\infty$.
\end{definition}

\begin{remark}\label{rem:protocol-respecting-distribution}
  As the probability of testing during the protocol is fixed, the expected frequency distributions will always take the form
\begin{equation}\label{eq:protocol-respecting-distribution-extended}
\p = \begin{pmatrix}
\gamma \w \\
(1-\gamma)
\end{pmatrix}
\end{equation}
for some $\w \in \QG$, where $p(\perp)$ is the final element of $\p$. Furthermore, the fixed testing probability means that any distribution that results in a finite infimum in~\eqref{eq:fmin-definition} necessarily takes this form. We shall refer to a distribution of the form~\eqref{eq:protocol-respecting-distribution-extended} as a \emph{protocol-respecting} distribution, denoting the set of all such distributions by $\respecting$.
\end{remark}

Particular properties of the min-tradeoff function that appear within the error terms of the EAT are:
\begin{itemize}
\item The maximum value attainable on $\P_{\C}$,
\begin{equation}
\Max{f} := \max_{\p \in \P_\C} f(\p).
\end{equation}
\item The minimum value over protocol-respecting distributions,
\begin{equation}
\Min{\frestricted} := \min_{\p \in \respecting} f(\p).
\end{equation}
\item The maximum variance over all protocol-respecting distributions,
\begin{equation}
\Var{\frestricted} := \max_{\p \in \respecting} \sum_{c \in \C} p(c) \left(f(\e_c)- f(\p)\right)^2.
\end{equation}
\end{itemize}

\begin{theorem}[EAT~\cite{DF}]
\label{thm:EAT}~\newline
Let $\{\N_i\}_{i=1}^n$ be a collection of EAT channels and let
$\rho_{\AA\BB\II\CC E} = \ptr{R_n}{(( \N_n \circ \dots \circ
  \N_1)\otimes \id_{E})\rho_{R_0 E}}$ be the output state after the
sequential application of the channels $\{\N_i\otimes\id_{E}\}_i$ to
some input state $\rho_{R_0 E}$. Let $\Omega \subseteq \mathcal{C}^n$
be some event that occurs with probability $p_{\Omega}$ and let
$\rho_{|_\Omega}$ be the state conditioned on $\Omega$ occurring. Finally let $\epsmooth \in (0,1)$ and $f$ be a valid min-tradeoff function for $\{\N_i\}_i$. 
If for all $\CC\in\Omega$, with $\pr{\CC}>0$, there is some $t \in \R$
for which $f(\freqcbm)\geq t$, then for any $\beta \in (0,1)$
\begin{equation}\label{eq:EAT-bound}
H_{\min}^{\epsmooth}(\AA\BB|\XX\YY\EE)_{\rho_{|_{\Omega}}} > n t - n (\errV + \errK) - \errW,
\end{equation}
where
\begin{equation}\label{eq:errV}
\errV := \frac{\beta \ln 2}{2} \left(\log\left(2 |\A\B|^2 +1\right) + \sqrt{  \Var{\frestricted} + 2} \right)^2,
\end{equation} 
\begin{equation}\label{eq:errK}
\errK := \frac{\beta^2}{6(1-\beta)^3\ln 2}\,
2^{\beta(\log |\A\B| + \Max{f} - \Min{\frestricted})}
\ln^3\left(2^{\log |\A\B| + \Max{f} - \Min{\frestricted}} + \expo^2\right)
\end{equation}
and
\begin{equation}\label{eq:errW}
\errW := \frac{1}{\beta}\left(1 - 2 \log(p_{\Omega}\,\epsmooth)\right).
\end{equation}
\end{theorem}

\begin{remark}\label{rem:alpha-choice}
As the EAT holds for all $\beta \in (0,1)$ we can numerically optimize
our choice of $\beta$ once we know the values of the other protocol
parameters. However, for large $n$ and small $\gamma$, a short
calculation shows that choosing $\beta \in O(\sqrt{\gamma/n})$ keeps
all the error terms of approximately the same magnitude. In
particular, this choice results in the error scalings: $n\errV\in O(\sqrt{n/\gamma})$, $n\errK\in O(1)$ and $\errW\in O(\sqrt{n/\gamma})$.
\end{remark}

\subsection{Randomness extractors}\label{sec:extractors}

Subject to the protocol not aborting, the entropy accumulation
sub-procedure detailed in \sec\ref{sec:accumulation-procedure} will
result in the production of some bit string $\AA\BB \in \{0,1\}^{2n}$ with $\hmin^{\epsmooth}(\AA\BB|\XX\YY\EE) > k$ for
some $k \in \R$. In order to `compress' this randomness into a shorter
but almost uniform random string a \emph{seeded, quantum-proof
  randomness extractor} can be used. This is a function
$\ext:\zraw\times\zseed\rightarrow\zout$, such that if $\zseed$ is a
uniformly distributed bit-string, the resultant bit-string
$\zout$ is $\epsilon$-close to uniformly distributed, even from the
perspective an adversary with quantum side-information $E$ about
$\zraw$. More formally, combining~\cite[Lemma~$3.5$]{DPVR} with the
standard definition for a quantum-proof randomness
extractor~\cite{KR2011} gives the following definition.
\begin{definition}[Quantum-proof strong extractor]\label{def:extractor}
A function
$\ext:\{0,1\}^{|\zraw|}\times\{0,1\}^{|\zseed|}\rightarrow\{0,1\}^{|\zout|}$
is a
\emph{quantum-proof $(k,\epext+2\epsmooth)$-strong extractor}, if for all
cq-states $\rho_{\AA\BB \EE}$ with $\hmin^{\epsmooth}(\AA\BB|\EE)_{\rho} \geq
k$ and for some $\epsmooth>0$ it maps $\rho_{\AA\BB \EE}\ot\tau_{\zseed}$
to $\rho'_{\ext(\AA\BB,\zseed)\zseed \EE}$ where
\begin{equation}
\frac12 \|\rho'_{\ext(\AA\BB,\zseed)\zseed \EE}-\tau_m\otimes \tau_{|\zseed|}\otimes\rho_\EE\|_1 \leq \epext + 2\epsmooth\,.
\end{equation}
(Recall that $\tau_m$ is the maximally mixed state on a system of dimension $m$.)
\end{definition}

Although in general the amount of randomness extracted will depend on
the extractor, $\hmin^{\epsmooth}(\AA\BB|\EE)$ provides an upper bound on
the total number of $\epsmooth$-close to uniform bits that can be
extracted from $\AA\BB$ and a well-chosen extractor will result in a
final output bit-string with
$|\zout|\approx\hmin^{\epsmooth}(\AA\BB|\EE)$. We denote any loss of
entropy incurred by the extractor by $\eploss=k-|\zout|$. Entropy
loss will differ between extractors but in general it will be some
function of the extractor error, the seed length and
the initial quantity of entropy. The extractor literature is rich with
explicit constructions, with many following Trevisan's
framework~\cite{trevisan}. For an in-depth overview of randomness
extraction, we refer the reader to~\cite{nisan1999extracting} and references therein.

\begin{remark}
  By using a \emph{strong} quantum-proof extractor, the output of the
  extractor will remain uncorrelated with the string used to seed
  it. Since the seed acts like a catalyst, we need not be overly
  concerned with the amount required. Furthermore, if available, it
  could just be acquired from a trusted public source immediately
  prior to extraction without compromising security. However, if a
  public source is used, it is important that it is not available to
  Eve too early in the protocol as this could allow Eve to create
  correlations between the outputs of the devices and the extractor
  seed.
\end{remark}

\begin{remark}
  Related to the previous remark is the question of whether the
  quantity we are interested in is
  $\hmin^{\epsmooth}(\AA\BB|\XX\YY\EE)$, rather than
  $\hmin^{\epsmooth}(\AA\BB|\EE)$ or
  $\hmin^{\epsmooth}(\AA\BB\XX\YY|\EE)$.  In common QKD protocols
  (such as BB84), the first of these is the only reasonable choice
  because the information $\XX\YY$ is communicated between the two
  parties over an insecure channel and hence could become known by
  Eve.  For randomness expansion, this is no longer the case: this
  communication can all be kept secure within one lab.  Whether the
  alternative quantities can be used then depends on where the seed
  randomness comes from.  If a trusted beacon is used then the first
  case is needed. If the seed randomness can be kept secure until such
  time that the random numbers need no longer be kept random then the
  second quantity could be used\footnote{This is a reasonable
    requirement, because there are other strings that have to be kept
    secure in the same way, e.g., the raw string $\AA$.}. If it is
  also desirable to extract as much randomness as possible, then the
  third quantity could be used instead. However, in many protocols the
  amount of seed required to choose $X$ and $Y$ in the entropy
  accumulation procedure is small enough that extracting randomness
  from this will not significantly increase the rate (see, e.g., our discussion in
  \app\ref{sec:input-rand}).
\end{remark}

\section{A template protocol for randomness expansion}\label{sec:adaptive-framework}

The primary purpose of this work is to provide a method whereby one
can construct tailored randomness expansion protocols, with a
guarantee of security and easily calculable generation rates. We
achieve this by providing a template protocol (Protocol~\protoName),
for which we have explicit security statements in terms of the
protocol parameters as well as the outputs of some SDPs. Our framework
is divided into three sub-procedures: preparation, accumulation
and extraction.

The preparation procedure consists of assigning values to the various
parameters of the protocol, this includes choosing a nonlocal game to
act as the nonlocality test. At the end of the preparation one would
have turned the protocol template into a precise protocol, constructed
a min-tradeoff function and be able to calculate the relevant security
quantities.  Note that once a specific protocol has been decided it is
not necessary to perform this step.  Furthermore, the manufacturer may
already specify the entire protocol to use with their devices, in which
case this step can be skipped.  Nevertheless, the fact that the
protocol can be tuned to the devices at hand enables the
user to optimize the randomness output from the devices at hand.

The final two parts of the framework
form the process described in Protocol~\protoName. The accumulation
step follows the entropy accumulation procedure detailed in
\sec\ref{sec:accumulation-procedure} wherein the user interacts with
their devices using the chosen protocol parameters. After the device
interaction phase has finished, the user implicitly evaluates the
quality of their devices by testing whether the observed inputs and
outputs satisfy the condition \eqref{eq:success-event}. Subject to the
protocol not aborting, a reliable lower bound on the min-entropy of
the total output string is calculated through the EAT
\eqref{eq:EAT-bound}. With this bound, the protocol can be completed
by applying an appropriate quantum-proof randomness extractor to the
devices' raw output strings.

\begin{figure}[t!]
\begin{center}
\fbox{
\begin{minipage}[c]{16cm}
\procedure[mode=text]{Protocol~\protoName}{%
\tbf{Parameters and notation}: \\
\ind 1 $\DD_{AB}$ -- a pair of untrusted devices taking
inputs from $\X$, $\Y$ and giving outputs from $\A$, $\B$ \\
\ind 1 $\G = (\mu,\V)$ -- a nonlocal game compatible with $\DD_{AB}$ \\
\ind 1 $\w \in \Q_{\G}$ -- an expected frequency distribution for $\G$ \\
\ind 1 $\d$ -- vector of confidence interval widths (satisfies $0\leq\delta_k\leq\omega_k$ for all $k \in [|\G|]$) \\
\ind 1 $n \in \mathbb N$ -- number of rounds \\
\ind 1 $\gamma \in (0,1)$ -- probability of a test round \\
\ind 1 $(\x,\y)$ -- distinguished inputs for generation rounds \\
\ind 1 $f_{\min}$ -- min-tradeoff function \\
\ind 1 $\epext > 0$ -- extractor error \\
\ind 1 $\epsmooth \in (0,1)$ -- smoothing parameter \\
\ind 1 $\epeat \in (0,1)$ -- entropy accumulation error \\
\ind 1 $\ext$ -- quantum-proof $(k,\epext+2\epsmooth)$-strong extractor \\
\ind 1 $\eploss$ -- entropy loss induced by $\ext$ \\
 [][\hline\hline]
\tbf{Procedure}: \\
\ind 1 \nln{1} Set $i=1$. \\
\ind 1 \nln{2} \tbf{While} $i \leq n$: \\
\ind 4 Choose $T_i=0$ with probability $1-\gamma$ and otherwise $T_1=1$. \\
\ind 4 \,\tbf{If} $T_i=0$: \\
\ind 6 \nln{Gen} Input $(\x,\y)$ into the respective devices, recording the inputs $X_iY_i$ and outputs $A_iB_i$.
\\
			   \ind 8 \,\,Set $\Ci= \perp$ and $i = i+1$. \\
\ind 4 \,\tbf{Else}: \\
\ind 6 \nln{Test} Play a single round of $\G$ on $\DD_{AB}$ using inputs sampled
			   from $\mu$, recording the inputs $X_iY_i$ and\\
			   \ind 8 \,\,outputs $A_i					   B_{i}$. Set $\Ci = \V(A_{i},B_{i},X_{i},Y_{i})$ and $i=i+1$.\\
%			   \ind 8 \,\,and $i=i+1$. \\
\ind 1 \nln{3} Compute the empirical frequency distribution $\freqcbm$. \\
\ind 2 \,\tbf{If} $ \gamma (\w-\d) < \freqcbm(\G) < \gamma (\w+\d)$: \\
\ind 4 \nln{Ext} Apply a strong quantum-proof randomness extractor $\ext$ to					   the raw output string $\AA\BB$\\
				\ind 6 producing $n f_{\min}(\w -\dpm) - \eploss$ bits $(\epext+2\epsmooth)$-close to uniformly distributed. \\
\ind 2 \,\tbf{Else}: \\
\ind 3 \,\nln{Abort} Abort the protocol.
}
\end{minipage}
}
\end{center}
\caption{The template quantum-secure device-independent randomness
  expansion protocol.}
\label{fig:full-protocol}
\end{figure} 

The next three subsections are dedicated to explaining these three sub-procedures in detail. In particular, \sec\ref{sec:subproc-prep} outlines the min-tradeoff function construction. A bound on the total entropy accumulated in terms of the various protocol parameters is then provided in \sec{\ref{sec:subproc-execution} and finally, in \sec\ref{sec:template-protocol}, the security statements for the template protocol are presented.

\subsection{Preparation}\label{sec:subproc-prep}

Before interacting with their devices the user must select appropriate
protocol parameters (see \fig\ref{fig:full-protocol} for a full list
of parameters). In particular, they must choose a nonlocal game to use
during the test rounds and construct a corresponding min-tradeoff
function.  This step enables this to be done if it is not already
specified.

The parameter values chosen will largely be dictated by situational
constraints; e.g., runtime, seed length and the expected
performance of the untrusted devices.\footnote{At first this may
  seem to conflict with the ethos of device-independence.  The point is that
  although the user of the protocol relies on an expected behaviour to
  set-up their devices, they do not rely on this expected behaviour
  being an accurate reflection of the devices for security. This also
  means that the expected behaviour could be that claimed by the
  device manufacturer. Using inaccurate estimation of the devices
  behaviour will not compromise security, but may lead to a
  different abort probability.} The user's choice of parameters, in
particular the choice of nonlocal game, will affect the form of their
min-tradeoff function derived and in turn their projected total
accumulated entropy. Before moving to the accumulation step of the
protocol the user can try to optimise their chosen parameters by computing the entropy rates for many different
choices. This allows them to adapt their protocol
to the projected performance of their devices.

We now present a constructible family of min-tradeoff functions for a
general instance of Protocol~\protoName. This construction is based on
the following idea. As noted in \sec\ref{sec:entropies-and-sdps} one
can numerically calculate a lower bound on the min-entropy of a system
based on its observed statistics. Pairing this with the relation,
$H_{\min}(X|E)\leq H(X|E)$, we have access to numerical bounds on the von Neumann
entropy. In particular, we can use the affine function
$g(\bm q)=\l\cdot \bm q$, where $\l$ is a feasible point
of the dual program~\eqref{prog:relaxed-dual}, in order to build a min-tradeoff function for the protocol.\footnote{In fact, by relaxing the dual program
  to the NPA hierarchy, the single round bound is valid against
  super-quantum adversaries. However, the full protocol is not
  necessarily secure more widely: to show that we would need to
  generalise the EAT and the extractor.} In order for $g$ to meet the requirements of a min-tradeoff function, its domain must be extended to include the symbol $\perp$. To perform this extension we use the method presented in~\cite[Section~5.1]{DF}. As the rounds are split into testing and generation rounds, we may decompose the EAT-channel for the $i^{\text{th}}$ round as $\N_i = \gamma\N^{\test}_i + (1-\gamma)\N_i^\gen$, where $\N_i^\test$ is the channel that would be applied if the round were a test round and $\N_i^{\gen}$ if the round were a generation round. Importantly, this splitting separates $\perp$ from the nonlocal game scores. That is, if $\N_i^{\test}$ is the channel applied then $\pr{\Ci = \perp} = 0$ whereas if $\N_i^{\gen}$ is applied then $\pr{\Ci=\perp} = 1$.

\begin{lemma}[Min-tradeoff extension~{\cite[Lemma~5.5]{DF}} ]\label{lem:extensionlemma}
Let $g: \P_{\G} \rightarrow \R$ be an affine function satisfying 
\begin{equation}\label{eq:entropy-bounding-OVG}
g(\p) \leq \inf_{\sigma_{R_{i-1}R'}: \N^{\test}_i(\sigma)_{\Ci}=\tau_{\p}} H(A_iB_i|X_iY_iR')_{\N_i(\sigma)}
\end{equation}
for all $\p\in\Q_{\G}$. Then, the function $f: \P_{\G\cup\{\perp\}}\rightarrow \R$, defined by its action on trivial distributions
\begin{align*}
f(\e_c) &= \Max{g} + \frac{g(\e_c) - \Max{g}}{\gamma}, \qquad \forall c \in \G, \\
f(\e_\perp) &= \Max{g},
\end{align*} 
is a min-tradeoff function for the EAT-channels $\{\N_i\}_i$. Furthermore, $f$ satisfies the following properties:
\begin{align*}
\Max{f} &= \Max{g}, \\
\Min{\frestricted} &\geq \Min{g}, \\
\Var{\frestricted} &\leq \frac{(\Max{g} - \Min{g})^2}{\gamma}.
\end{align*}
\end{lemma}

\begin{lemma}[Min-tradeoff construction]\label{lem:fmin}
Let $\G$ be a nonlocal game and $k \in \mathbb N$. For each $\v\in\Qk_\G$,
let $ \lv$ be a feasible point of
\prog\eqref{prog:relaxed-dual} when parameterized by
$\v$. Furthermore, let $\lmax =
\max_{c \in \G}\lambda_{\v}(c)$ and $\lmin =
\min_{c \in \G}\lambda_{\v}(c)$. Then, for any set of EAT channels
$\{\N_i\}_{i=1}^n$ implementing an instance of Protocol~\protoName\ with the nonlocal game $\G$, the set of functionals $\Fmin(\G) = \{f_{\v}(\cdot)\mid \v\in \Qk_\G\}$ forms a family of min-tradeoff functions, where $f_{\v}: \P_\C \rightarrow \R$ are defined by their actions on trivial distributions
\begin{align}\label{eq:fmin-c}
\hspace{4cm}f_{\v}(\e_{c}) &:= (1-\gamma)\left(\Av - \Bv \frac{\lv\cdot\e_c - (1-\gamma)\lmin}{\gamma} \right) \hspace{0.5cm}\text{for } c \in \G,\intertext{and}
f_{\v}(\e_\perp) &:= (1-\gamma)\left(\Av - \Bv\, \lmin \right), &~
\end{align}
where $\Av = \tfrac{1}{\ln 2} - \log (\lv\cdot\v)$ and $\Bv = \frac{1}{\lv\cdot\v \ln 2}$.

Moreover, these min-tradeoff functions satisfy the following relations.
\begin{itemize}
\item Maximum:
\begin{equation}\label{eq:fv-max}
\Max{f_{\v}} = (1-\gamma) (\Av - \Bv \,\lmin)
\end{equation}
\item $\respecting$-Minimum:
\begin{equation}\label{eq:fv-min}
\Min{\left.f_{\v}\right|_{\respecting}} \geq (1-\gamma)(\Av - \Bv \,\lmax)
\end{equation}
\item $\respecting$-Variance:
\begin{equation}\label{eq:fv-var}
\Var{\left.f_{\v}\right|_{\respecting}} \leq \frac{(1-\gamma)^2 \Bv^2 (\lmax - \lmin)^2}{\gamma}
\end{equation}
\end{itemize}
\begin{proof}
Consider the entropy bounding property \eqref{eq:entropy-bounding-OVG} but with $\C$ restricted to the scoring alphabet of $\G$, i.e., we have an affine function $g_{\v}: \P_\G \rightarrow \mathbb{R}$ such that
$$
g_{\v}(\bm q) \leq \inf_{\sigma_{R_{i-1}R'}: \N^{\test}_i(\sigma)_{\Ci(\G)}=\tau_{\bm q}} H(A_iB_i|X_iY_iR')_{\N_i(\sigma)},
$$
for all $\bm q \in \Q_\G$. 

As conditioning on additional side information will not increase the von Neumann entropy, we may condition on whether or not the round was a test round,

\begin{align*}
H(A_iB_i|X_iY_iR')_{\N_i(\sigma)} &\geq H(A_iB_i|X_iY_i T_i R')_{\N_i(\sigma)} \\
&= \gamma H(A_iB_i|X_iY_i,T_i=1,R')_{\N_i(\sigma)} + (1-\gamma) H(A_iB_i|X_iY_i, T_i = 0, R')_{\N_i(\sigma)} \\
& > (1-\gamma) H(A_iB_i|X_i=\x,Y_i=\y, T_i = 0, R')_{\N_i(\sigma)}
\end{align*}
where in the final line we have used the fact that the inputs are fixed for generation rounds. As the min-entropy lower bounds the von Neumann entropy, we arrive at the bound
$$
H(A_iB_i|X_iY_iR')_{\N_i(\sigma)} > (1-\gamma) \hmin(A_iB_i|X_i=\x,Y_i=\y, T_i = 0, R')_{\N_i(\sigma)}.
$$
Using programs~\eqref{prog:relaxed-primal} and~\eqref{prog:relaxed-dual}, we can lower bound the right-hand side in terms of the relaxed guessing probability. Specifically, for a single generation round
\begin{align*}
H_{\min}(AB|X=\x, Y=\y, T=0, R') &= -\log(\pg(\bm q)) \\
& \geq -\log(\lv^{(k)} \cdot \bm q), 
\end{align*}
holds for all $k\in\mathbb N$, any $\v \in \Qk_\G$ and any quantum
system realising the expected statistics $\bm{q} \in \Q_\G$. In the
final line we used the monotonicity of the logarithm together with the
fact that a solution to the relaxed dual program, for any
parameterization $\v \in \Qk_\G$, provides a linear function $ \lv
\cdot (\,\cdot\,)$ that is greater than $\pg$ everywhere on
$\Qk_{\G}$. Note that this bound is also device-independent and is
therefore automatically a bound on the infimum. Dropping the
superscript $(k)$ for notational ease, we may recover the desired affine property by taking a first order expansion about the point $\v$. This results in the function
$$
g_{\v}(\bm q):= (1-\gamma) (\Av - \Bv \, \lv \cdot \bm q),
$$
which satisfies
$$
g_{\v}(\bm q) \leq \inf_{\sigma_{R_{i-1}R'}: \N^{\test}_i(\sigma)_{\Ci}=\tau_{\bm q}} H(A_iB_i|X_iY_iR')_{\N_i(\sigma)},
$$
for all $\bm q \in \Q_\G$, with $\Av$ and $\Bv$ as defined in Lemma~\ref{lem:fmin}. The statement then follows by applying Lemma~\ref{lem:extensionlemma} to $g_{\v}$, noting $\Max{g_{\v}} = (1-\gamma) (\Av - \Bv \,\lmin)$ and $\Min{g_{\v}} = (1-\gamma) (\Av - \Bv \,\lmax)$.
\end{proof}
\end{lemma}

\begin{example}\label{ex:prep1}
	Taking the nonlocal game $\G_{\chsh}$ introduced in Example~\ref{ex:chsh}, we can use the above lemma to construct a min-tradeoff function. Fixing the probability of testing, $\gamma = 5 \times 10^{-3}$, we consider a device that behaves (during a test round) according to the expected frequency distribution~$\w = (\wcr, \wch, 1-\wcr-\wch)$. In \fig\ref{fig:prep-example1}, we plot the certifiable min-entropy of a single generation round for a range of $\w$. We see that as the scores approach $\w = \tfrac12\,\left(1,\tfrac{2 + \sqrt{2}}{4}, \tfrac{2 - \sqrt{2}}{4} \right)$, we are able to certify almost\footnote{Due to the infrequent testing we are actually only able to certify a maximum of $2\cdot (1-\gamma)$ bits per interaction.} two bits of randomness per entangled qubit pair using $\G_{\chsh}$. 
	\begin{figure}[t!]
		\begin{center}
			\includegraphics[scale=1.0]{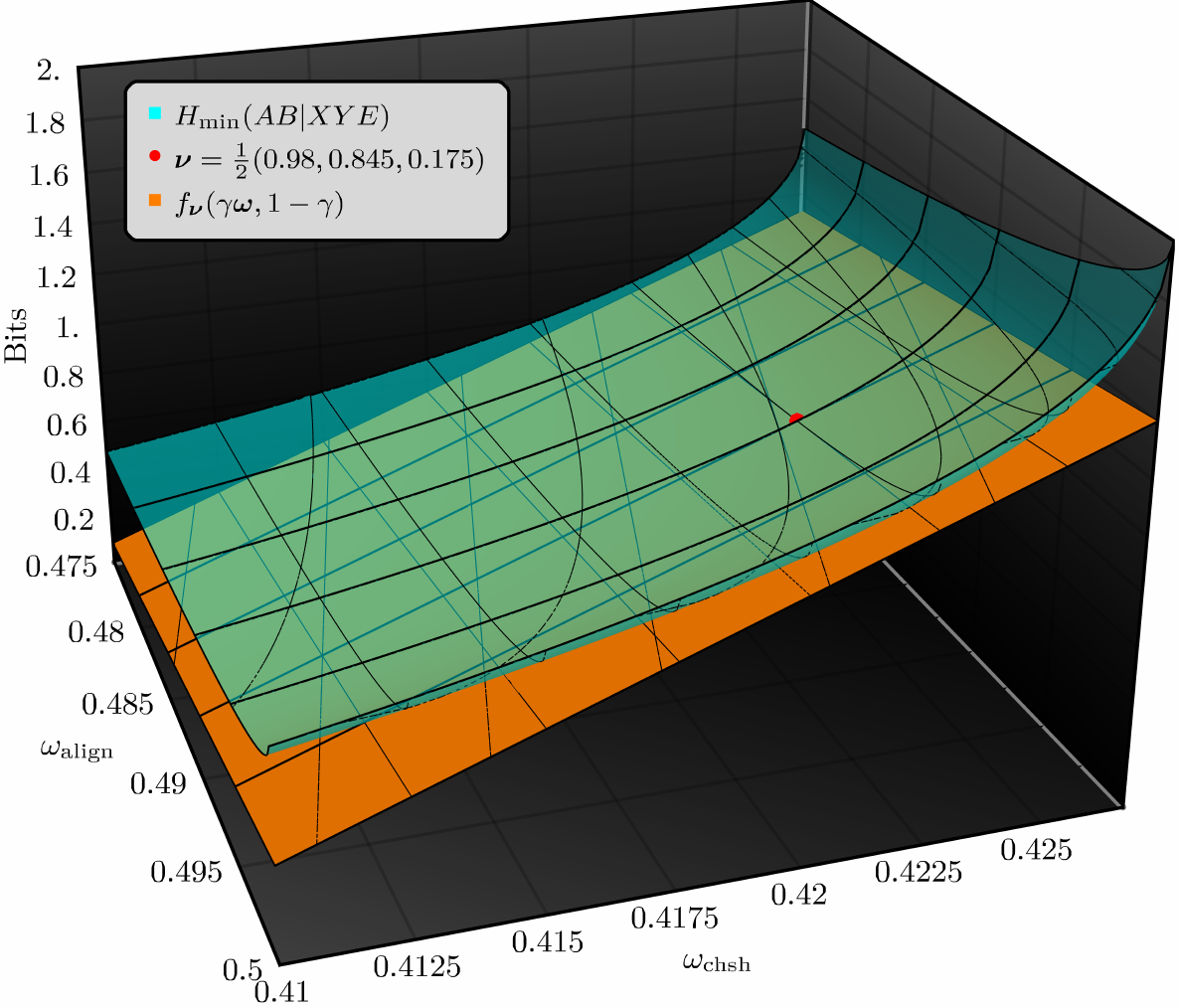}
		\end{center}
		\caption[Lower bounds to min-entropy surface]{A plot of a lower bound on the certifiable min-entropy produced during a single round of the protocol. This lower bound was calculated using \prog\ref{prog:relaxed-primal} relaxed to the second level of the NPA hierarchy. In addition, we plot a min-tradeoff function $f_{\v}$ evaluated for distributions of the form $\p = (\gamma \w, 1-\gamma)$ for $\w \in \Q_{\G}$, i.e. expected frequency distributions over $\G\cup\{\perp\}$ that are compatible with the spot-checking structure of the rounds. Since $f_{\v}$ is the tangent plane to the surface at the point $\v$ it forms an affine lower bound on the min-entropy of any quantum distribution compatible with the protocol.}
		\label{fig:prep-example1}
	\end{figure}
\end{example}

\subsection{Accumulation and extraction}\label{sec:subproc-execution}
After fixing the parameters of the protocol and constructing a
min-tradeoff function $f_{\min}$, the user proceeds with the remaining
steps of Protocol~\protoName: accumulation and extraction. The
accumulation step consists of the device interaction and evaluation
sub-procedures that were detailed in
\sec\ref{sec:accumulation-procedure}. If the protocol does not abort,
then with high probability the generated string $\AA\BB$ contains at
least some given quantity of smooth min-entropy. The following lemma
applies the EAT to deduce a lower bound on the amount of entropy
accumulated.

\begin{lemma}[Accumulated entropy]\label{lem:accumulated-entropy}
Let the randomness expansion procedure and all of its parameters be as
defined in \fig\ref{fig:full-protocol}. Furthermore, let $\Omega$ be
the event the protocol does not abort (cf.~\eqref{eq:success-event})
and let $\rho_{|\Omega}$ be the final state of the system conditioned
on this. Then, for any $\beta,\,\epsmooth,\,\epeat\in(0,1)$ and any
choice of min-tradeoff function $f_{\v}\in \Fmin$, either
Protocol~\protoName\ aborts with probability greater than $1-\epeat$ or 
\begin{equation}\label{eq:EAT-total}
H_{\min}^{\epsmooth}(\AA\BB|\XX\YY\EE)_{\rho_{|_{\Omega}}} >  (1-\gamma) n \left(\Av - \Bv\lv \cdot (\w -\dpm)\right)  - n (\errV + \errK) - \errW,
\end{equation}
where
\begin{equation}\label{eq:errV-explicit}
\errV := \frac{\beta \ln 2}{2} \left(\log\left(2 |\A\B|^2 +1\right) + \sqrt{  \frac{(1-\gamma)^2 \Bv^2 (\lmax - \lmin)^2}{\gamma} + 2} \right)^2,
\end{equation} 
\begin{equation}\label{eq:errK-explicit}
\errK := \frac{\beta^2}{6(1-\beta)^3\ln 2}\,
2^{\beta(\log |\A\B| + (1-\gamma)\Bv (\lmax-\lmin))}
\ln^3\left(2^{\log |\A\B| + (1-\gamma)\Bv (\lmax-\lmin)} + e^2\right),
\end{equation}
\begin{equation}\label{eq:errW-explicit}
\errW := \frac{1}{\beta}\left(1 - 2 \log(\epeat\,\epsmooth)\right)
\end{equation}
and $\dpm = (\delta(c)\,\sgn(-{\lambda_{\v}(c)}))_{c \in
  \G}$.
\begin{proof}
  Let $\{\N_i\}_{i\in [n]}$ be the set of channels implementing the
  entropy accumulation sub-procedure of Protocol~\protoName. Comparing
  this procedure with the definition of the EAT channels
  \defn\ref{def:eat-channels}, we have
  $\N_i:\mathcal{S}(R_{i-1}) \rightarrow
  \mathcal{S}(A_iB_iX_iY_iT_iC_iR_i)$ with
  $A_i, B_i, X_i, Y_i, T_i, C_i$ finite dimensional classical systems,
  $R_i$ an arbitrary quantum system and the score $C_i$ is a
  deterministic function of the values of the other classical
  systems. Furthermore, the inputs to the protocol for the
  $i^{\text{th}}$ round, $(X_i,Y_i,T_i)$, are chosen independently of
  all other systems in the protocol and so the conditional
  independence constraints
  $I(A_1^{i-1}B_1^{i-1}\!\! :\! X_i Y_i | X_1^{i-1}Y_1^{i-1} E) = 0$
  hold trivially. The conditions necessary for $\{\N_i\}_{i\in [n]}$
  to be EAT-channels are satisfied and by \lem\ref{lem:fmin} $f_{\v}$
  is a min-tradeoff function for these channels. We can now apply the
  EAT to bound the total entropy accumulated.

Consider now the pass probability of the protocol,
$p_{\Omega}$. Either $p_\Omega < \epeat$, in which case the protocol
will abort with probability at least $1 - \epeat$, or $p_\Omega \geq
\epeat$. In the latter case we can replace the unknown $p_{\Omega}$
in~\eqref{eq:errW} with $\epeat$ as this results in an increase in the error term $\errW$. The EAT then asserts that
$$
\hmin^{\epsmooth}(\AA\BB|\XX\YY\EE)_{\rho_{|\Omega}}
>
n \inf_{\CC\in\Omega} f_{\v}(\freqcbm) - n (\errV + \errK) - \errW,
$$
for any choice of min-tradeoff function $f_{\v} \in \Fmin$.

As the min-tradeoff functions are affine, we can lower bound the infimum for the region of possible scores specified by the success event,
$$
\Omega = \left\{ \CC \mid \gamma (\w-\d) < \freqcbm(\G) < \gamma (\w+\d) \right\}.
$$
Taking $\p  = (\gamma (\w - \dpm), (1-\gamma)$), we
have $f(\p) \leq\inf_{\CC\in\Omega} f_{\v}(\freqcbm)$. Note that $\p$ may not correspond to a frequency distribution that could have resulted from a successful run of the protocol -- it may not even be a probability distribution. However, it is sufficient for our purposes as an explicit lower bound on the infimum. Further, noting that $f_{\v}(\p) = g_{\v}(\w - \dpm)$, we can straightforwardly compute this lower bound as 
$$
f_{\v}(\p) = (1-\gamma) \left(\Av - \Bv\lv \cdot (\w -\dpm)\right). 
$$
Inserting the min-tradeoff function
properties~\eqref{eq:fv-max}--\eqref{eq:fv-var} into the the EAT's
error terms [\eqref{eq:errV}--\eqref{eq:errW}] we get the explicit
form of the quantities $\errV$, $\errK$ and $\errW$ stated in the lemma.
\end{proof}
\end{lemma}

If the protocol does not abort during the accumulation procedure, the
user may proceed by applying a quantum-proof strong extractor to the
concatenated output string $\AA\BB$ resulting in a close to uniform bit-string of length approximately $(1-\gamma) n \left(\Av - \Bv\lv \cdot (\w -\dpm)\right)  - n (\errV + \errK ) - \errW$.

	\begin{figure}[t!]
		\begin{center}
			\includegraphics{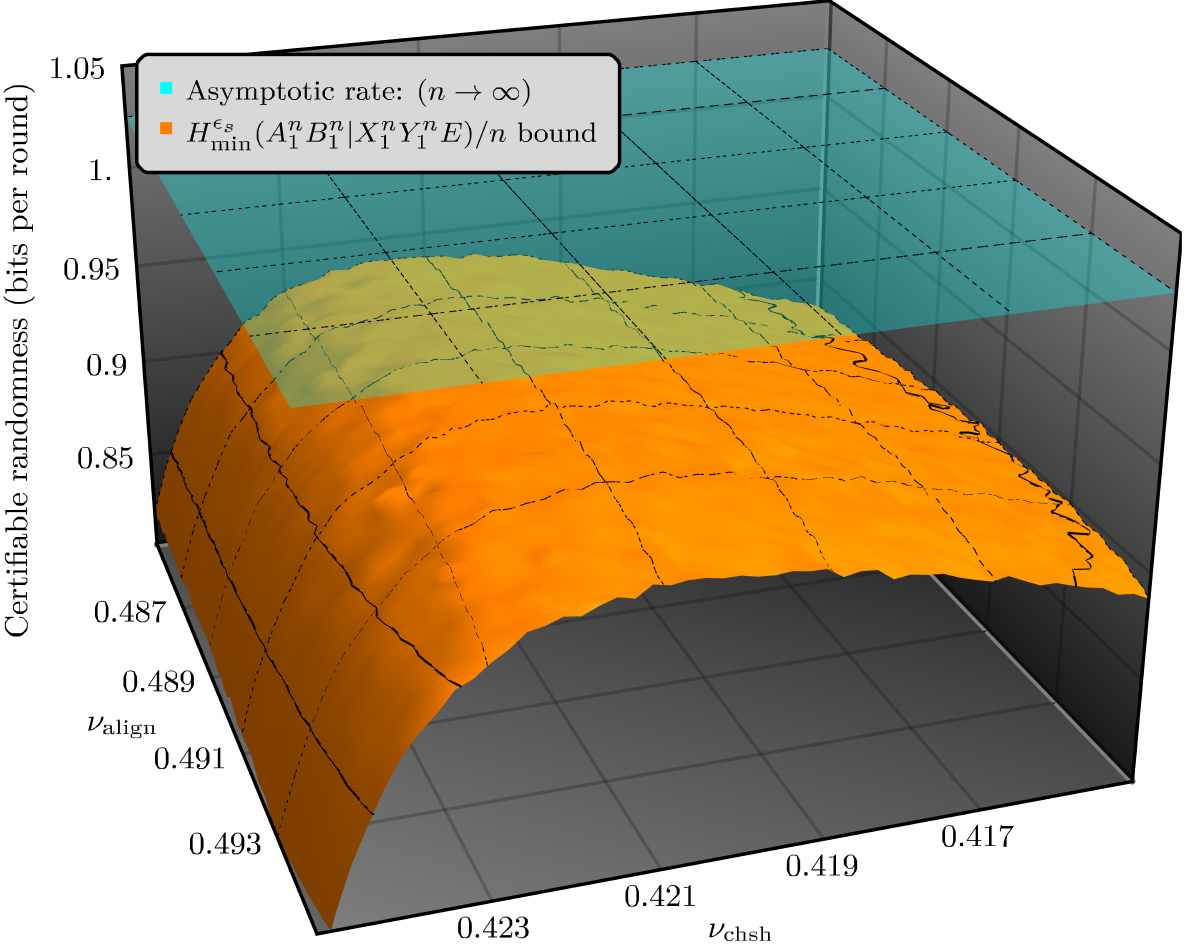}
		\end{center}
		\caption{A plot of the randomness certified as we vary
                  our choice of min-tradeoff function. At each point
                  $\nu$ we evaluate the certifiable randomness
                  \eqref{eq:EAT-total} for the corresponding choice of
                  min-tradeoff function $f_{\nu}$, numerically
                  optimizing the parameter $\beta$ each time. The
                  rough appearance of the surface is a result of finding local optima in the
                  $\beta$ optimization. For reference, we include a
                  plot of the asymptotic rate, i.e.,
                  \eqref{eq:EAT-total} as $n \rightarrow \infty$ and
                  $\d \rightarrow \bm 0$. The protocol parameters used
                  during the calculations are as follows: $n =
                  10^{10}$, $\gamma=  5 \times 10^{-3}$, $\w =(0.49,0.4225,0.0875)$,
                  $\delta_{\chsh}=\delta_{\cross} = 10^{-3}$ and
                  $\epsmooth=\epeat=10^{-8}$.}
		\label{fig:prep-example2}
	\end{figure}
\begin{example}\label{ex:prep2}
	
	 Continuing from \ex\ref{ex:prep1}, we look at the bound on
         the accumulated entropy specified by \eqref{eq:EAT-total} for
         a range of choices of $f_{\v} \in \Fmin$. Again, we are
         considering a quantum implementation with an expected frequency distribution $\w =(0.49,0.4225,0.0875)$. In \fig\ref{fig:prep-example2} we
         see that our choice of min-tradeoff function can have a large
         impact on the quantity of entropy we are able to certify. The plot gives some reassuring numerical evidence that, for the nonlocal game $\G_{\chsh}$, the certifiable randomness is continuous and concave in the family parameter $\v$. 

The min-tradeoff function indexed by our expected frequency distribution,
$f_{\w}$, is able to certify just under $0.939$-bits per
interaction. By applying a gradient-ascent algorithm we
were able to improve this to $0.946$-bits per interaction. In an attempt to avoid getting stuck within local optima we
applied the algorithm several times, starting subsequent iterations at randomly chosen points close to the current optimum.
The optimization led to an improved min-tradeoff function choice $f_{\v^*}$, where $\v^* = (0.491,0.421,0.088)$.
\end{example}

\subsection{Protocol~\protoName}\label{sec:template-protocol}
Protocol~\protoName\ is the concatenation of the accumulation and
extraction sub-procedures. It remains to provide the formal security
statements for a general instance of Protocol~\protoName. We refer to
an untrusted device network $\DD_{AB}$ as \emph{honest} if during each
interaction, the underlying quantum state shared amongst the devices
and the measurements performed in response to inputs remain the same
(i.e., the devices behave as the user expects). Furthermore, each
interaction is performed independently of all others. The following
lemma provides a bound on the probability that an honest
implementation of Protocol \protoName\ aborts.

\begin{lemma}[Completeness of Protocol~\protoName]
\label{lem:complet}
Let Protocol~\protoName\ and all of its parameters be as defined in
\fig\ref{fig:full-protocol}. Then, the probability that an honest
implementation of Protocol~\protoName\ aborts is no greater than $\epcomp$ where 
\begin{equation}\label{eq:completeness}
\epcomp=2\sum_{k=1}^{|\G|}\expo^{-\frac{\gamma\delta_k^2}{3\omega_k}n}. 
\end{equation}

\begin{proof}
During the parameter estimation step of Protocol~\protoName, the protocol aborts if the observed frequency distribution $\freqcbm$ fails to satisfy
\begin{equation*}
\gamma (\w-\d) < \freqcbm(\G) < \gamma (\w+\d).
\end{equation*}
Writing $\freqcbm(\G) = (r_k)_{k=1}^{|\G|}$, $\w=(\omega_k)_{k=1}^{|\G|}$ and $\d=(\delta_k)_{k=1}^{|\G|}$, the probability that an honest implementation of the protocol aborts can be written as  
\begin{align*}
\pabort = \pr{ \bigcup_{k=1}^{|\G|} \bigg\{\big|r_k - \gamma \omega_k \big| \geq \gamma \delta_k\bigg\} } \leq \sum_{k=1}^{|\G|} \pr{\big|r_k - \gamma \omega_k \big| \geq \gamma \delta_k }.
\end{align*}

Restricting to a single element $r_k$ of $\freqcbm(\G)$, we
can model its final value as the binomially distributed random
variable $r_k \sim \tfrac{1}{n}\bin{n}{\gamma \omega_k}$. As a consequence of the
Chernoff bound (cf.\ Corollary~\ref{cor:Chernoff}), and that
$\delta_k<\omega_k$, we have 
$$\pr{\big|r_k - \gamma \omega_k \big| \geq \gamma \delta_k }\leq2\expo^{-\frac{\gamma \delta_k^2 n}{3 \omega_k}}.$$
Applying this bound to each element of the sum individually, we arrive at the desired result.
\end{proof}
\end{lemma}

\begin{remark}
\label{rem:complet} 
The completeness error in the above lemma only considers the
possibility of the protocol aborting during the parameter estimation
stage. However, if the initial random seed is a particularly limited
resource then it is possible that the protocol aborts due to seed
exhaustion. In \lem\ref{lem:input-randomness} we analyse a sampling
algorithm required to select the inputs during device interaction. If
required, the probability of failure for that algorithm could be
incorporated into the completeness error.
\end{remark}

With a secure bound on the quantity of accumulated entropy established by \lem\ref{lem:accumulated-entropy} we can apply a $(k,\epext + 2 \epsmooth)$-strong extractor to $\AA\BB$ to complete the security analysis. Combined with the input randomness discussed in \app\ref{sec:input-rand} we arrive at the following theorem. 

\begin{lemma}[Soundness of Protocol~\protoName]
\label{lem:sound}
Let Protocol~\protoName\ be implemented with some initial random seed
$D$ of length $d$. Furthermore let all other protocol parameters be
chosen within their permitted ranges, as detailed in
\fig\ref{fig:full-protocol}. Then the soundness error of
Protocol~\protoName\ is
$$\epsound=\max(\epext+2\epsmooth,\epeat)\,.$$
\end{lemma}
\begin{proof}
  Recall from~\eqref{eq:sound_def} that the soundness error is an
  upper bound on
  $\frac{1}{2}\mathrm{Pr}[\Omega]\cdot\|\rho_{ZE}-\tau_m\ot\rho_E\|_1$. In
  the case $\pr{\Omega}\leq\epeat$, we have $\frac{1}{2}\mathrm{Pr}[\Omega]\cdot\|\rho_{ZE}-\tau_m\ot\rho_E\|_1\leq\epeat$.

  In the case $\pr{\Omega}>\epeat$,
  \lem\ref{lem:accumulated-entropy} gives a bound on the accumulated
  entropy. Combining with the definition of a quantum-proof strong
  extractor \defn\ref{def:extractor} and noting that the norm
  is non-increasing under partial trace we obtain
  $\frac{1}{2}\mathrm{Pr}[\Omega]\cdot\|\rho_{ZE}-\tau_m\ot\rho_E\|_1\leq\epext+2\epsmooth$,
  from which the claim follows.
\end{proof}
\begin{remark}
  By choosing parameters such that $\epeat\leq\epext+2\epsmooth$ we
  can take the soundness error to be $\epext+2\epsmooth$.
\end{remark}

Combining all of the previous results we arrive at the full security statement concerning Protocol~\protoName. 

\begin{theorem}[Security of Protocol~\protoName]\label{thm:security}
Protocol~\protoName\ is an $(\epcomp,\epsound)$-secure randomness expansion protocol producing 
\begin{equation}\label{eq:security-of-qre}
((1-\gamma) \left(\Av - \Bv\lv \cdot (\w -\dpm)\right)  - \errV - \errK)\, n - \errW- \eploss
\end{equation}
random bits at least $\epsound$-close to uniformly distributed, where $\epcomp$, $\epsound$ are given by \lem\ref{lem:complet} (cf. \rem\ref{rem:complet}) and \lem\ref{lem:sound}.
\end{theorem}

\begin{remark}
	The expected seed length required to execute
        Protocol~\protoName\ is $d \approx \left(\gamma H(\mu) + h(\gamma)\right) n$ (cf.\ \lem\ref{lem:input-randomness}).
\end{remark}
\begin{example}
In \ex\ref{ex:prep1} and \ex\ref{ex:prep2} we used the following choice of protocol parameters: $n = 10^{10}$, $\gamma = 5\times 10^{-3}$, $\delta_1 = \dots = \delta_{|\G|} = 10^{-3}$ and $\epsmooth = \epeat = 10^{-8}$. The resulting implementation of Protocol~\protoName, using the nonlocal game $\G_{\chsh}$ with an expected frequency distribution $\w =(0.49,0.4225,0.0875)$, exhibits the following statistics.

\begin{center}
\begin{tabular}{|l|l|}
\hline
\tbf{Quantity}                                             & \tbf{Value}         \\ \hline                 
Total accumulated entropy before extraction (no abort) & $9.46 \times 10^{9}$           \\ %\hline
Expected length of required seed before extraction     & $5.54 \times 10^{8}$            \\ %\hline
Expected net-gain in entropy (no abort)               & $8.91 \times 10^{9} - \eploss$ \\ %\hline
Completeness error ($\epcomp$)                        & $8.77 \times 10^{-8}$      
\\ \hline
\end{tabular}
\end{center}
\end{example}

\section{Examples}\label{sec:examples}
In this section we demonstrate the use of our framework through the
construction and analysis of several protocols based on different tests of nonlocality. 
To this end, we begin by introducing two families of nonlocal games which we consider alongside $\G_{\chsh}$. 

\nl\noindent\textbf{Empirical behaviour game $(\G_{\FB})$.} 
The \emph{empirical behaviour game} $(\G_{\FB})$ is a nonlocal game that estimates the underlying behaviour of $\DD_{AB}$, i.e., it attempts to characterise each
individual probability $\pabgxy$. We may construct this by associating with each input-output tuple $(a,b,x,y)\in\ABXY$ a corresponding score $c_{abxy} \in \G$ and defining the scoring rule
$$
V_{\FB}(a,b,x,y) := c_{abxy},
$$ 
for each $(a,b,x,y) \in \ABXY$. Then, for any input distribution $\mu_{\FB}$ with full support
on the alphabets $\X\Y$, the collection $\G_{\FB} = (\mu_{\FB},V_{\FB})$
forms a nonlocal game. Moreover, for agents playing according to some strategy $\p \in \Q$,
the expected frequency distribution over the scores is precisely the joint distribution,
\begin{align*}
\omega_{\FB}(a,b,x,y) &= \mu_{\FB}(x,y) p(a,b|x,y) \\
&= p(a,b,x,y).
\end{align*}
As $\G_{\FB}$ can be defined for any collection of input-output
alphabets, we can indicate the size of these alphabets as
superscripts, i.e., $\G_{\FB}^{|\X||\Y||\A||\B|}$. However, since we
only consider binary output alphabets in this work, we will not
include their sizes in the superscript, i.e., we will write
$\G^{23}_{\FB}$ instead of $\G^{2322}_{\FB}$.

\begin{remark}\label{rem:redundancies}
	The scoring rule for $\G_{\FB}$, as defined above, has several redundant components, arising from normalisation and the no-signalling conditions. In fact, there are only $[(|\A|-1)|\X| + 1][(|\B|-1)|\Y| +1] -1$ free parameters~\cite{Cirelson93}. Knowing this we can reduce the number of scores in our nonlocal game and, in turn, the number of constraints we impose in our SDPs.\footnote{It is important to remove redundant constraints in practice as they can lead to numerical instabilities.}
\end{remark}

~\\
\noindent
\textbf{Joint correlators game $(\G_{\al})$.}
Specifically, for each $(x,y) \in \X\Y$ we define a score $c_{xy}$ and a scoring rule  
$$
V_{\al}(a,b,x,y) := \begin{cases}
c_{xy} \qquad &\text{if } a=b \\
\cnorm \qquad &\text{otherwise}.
\end{cases}
$$ 
That is, for a pair of inputs $(x,y)$ the score is recorded as $c_{xy}$ whenever
the agents' outcomes agree. Otherwise, they record some normalization
score $\cnorm$. The input distribution can then be specified in some way: we use the uniform distribution over $\X\Y$. We refer to this game by the symbol $\G_{\al}$ and, as before, we will indicate the sizes of the input alphabets with superscripts.

\begin{figure}[t]
\begin{center}
	\begin{subfigure}{0.48\textwidth}
		\includegraphics[trim= 5pt 5pt 5pt 5pt,clip=true,width=\textwidth]{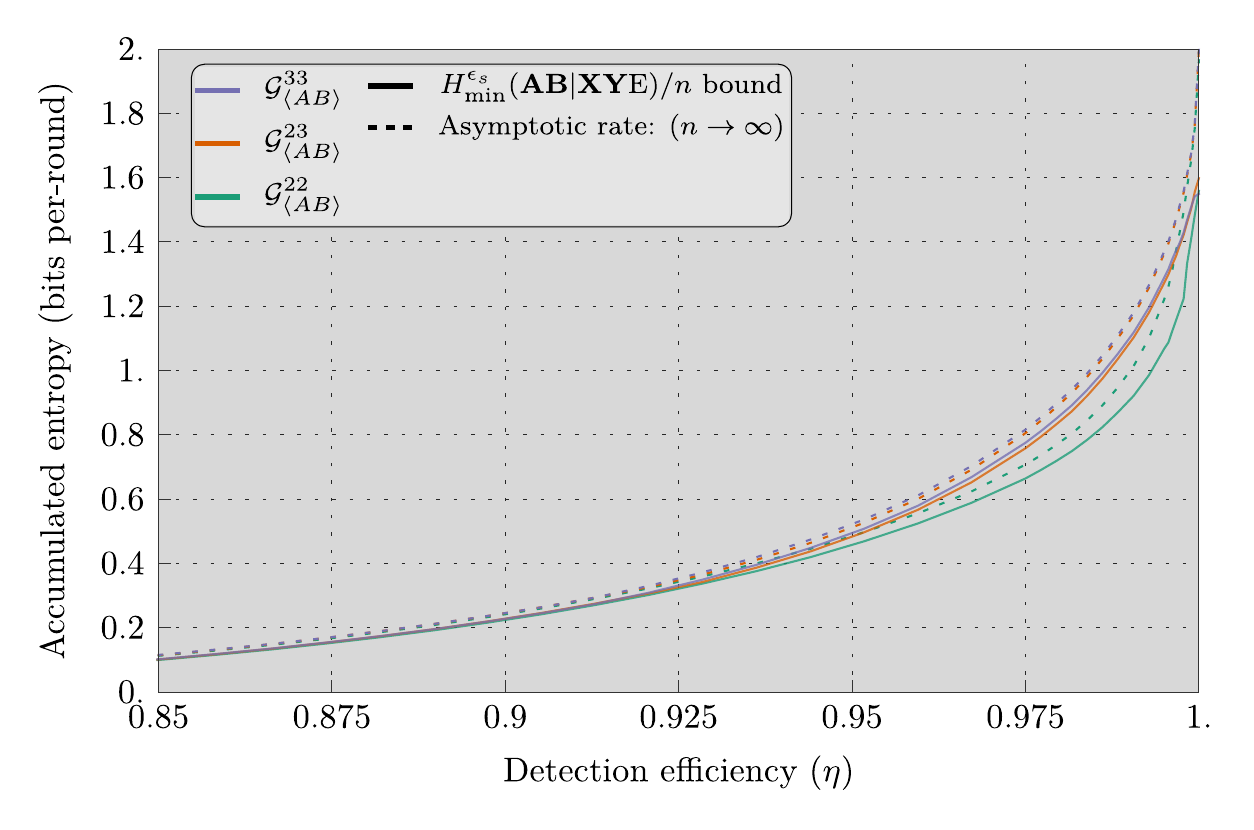}
			\caption{$\G_{\al}$ protocols}\vspace{-0.2cm}\label{subfig:corrplot}
	\end{subfigure}
	\begin{subfigure}{0.48\textwidth}
		\includegraphics[trim= 5pt 5pt 5pt 5pt,clip=true,width=\textwidth]{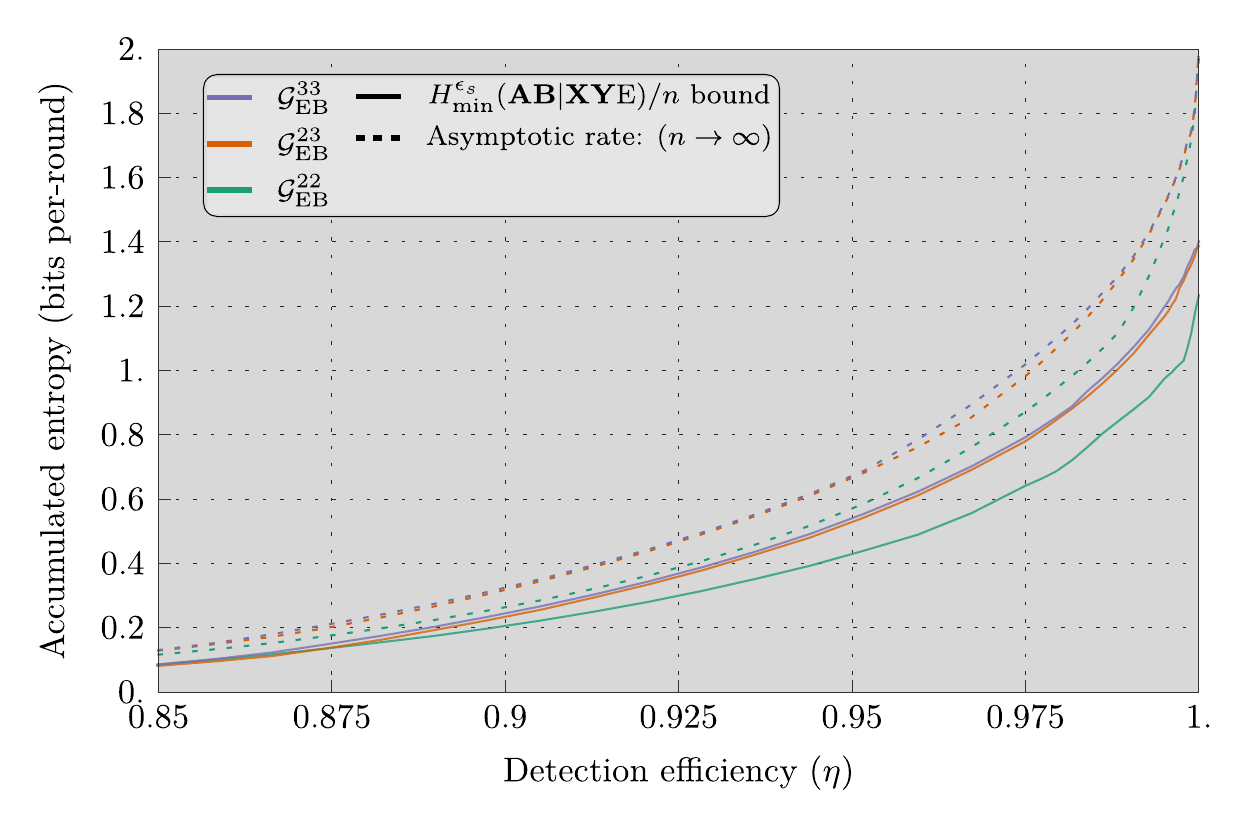}
				\caption{$\G_{\FB}$ protocols}\vspace{-0.2cm}\label{subfig:ebplot}
	\end{subfigure}
	\begin{subfigure}{\textwidth}
	\centering\vspace{0.25cm}
		\includegraphics[trim= 5pt 5pt 5pt 5pt,clip=true]{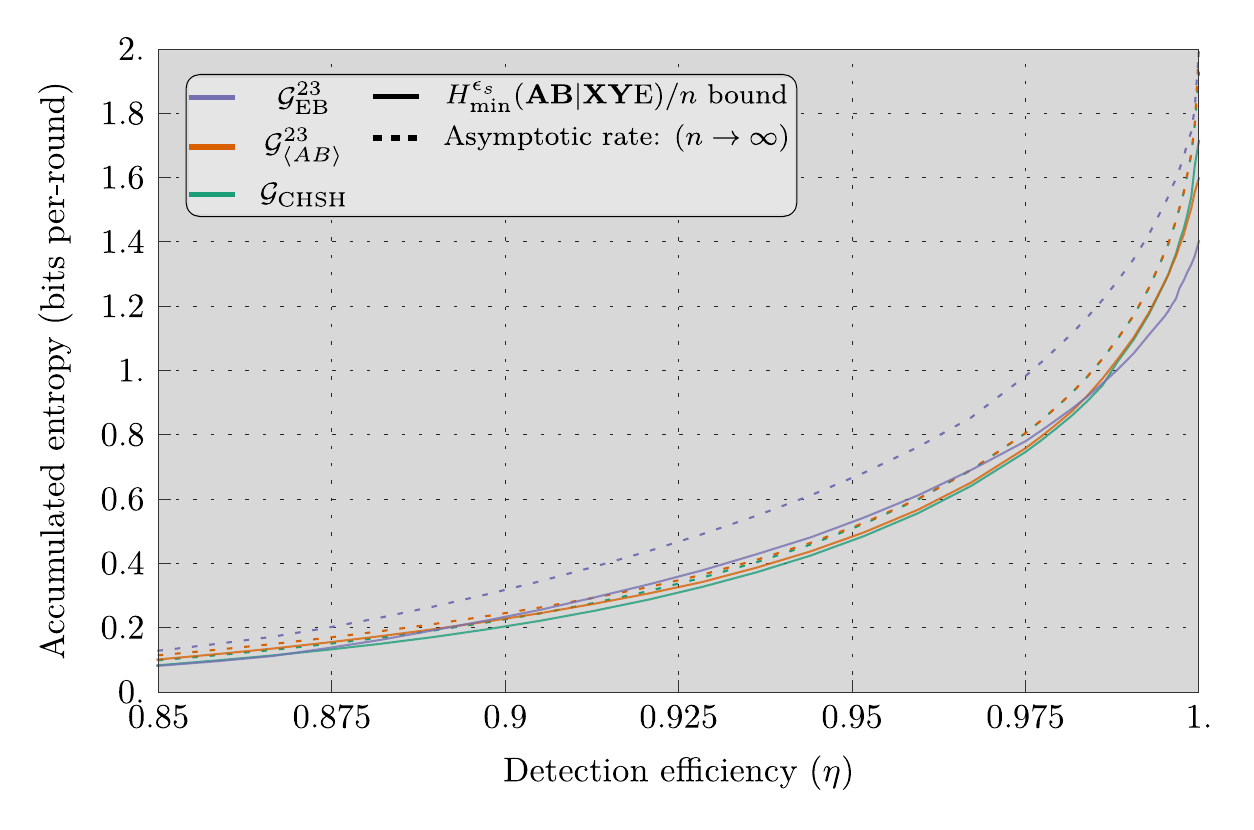}
					\caption{Comparison of protocols in the $(2,3)$-scenario.}\vspace{-0.25cm}			\label{subfig:23comparison}
	\end{subfigure}
\caption{A plot of the asymptotic and EAT-rates for protocols using the nonlocal game families $\G_{\al}$, $\G_{\FB}$ and $\G_{\chsh}$.}\label{fig:detection-plots}
\end{center}
\end{figure}

\subsection{Rates in the presence of inefficient detectors}\label{sec:noise-robustness-plots}
We now compare the accumulation rates of protocols built using the nonlocal games described above.  We retain the protocol parameter choices from the previous examples: $n=10^{10}$, $\gamma = 5 \times 10^{-3}$ and $\epsmooth = \epeat = 10^{-8}$, except we now set the confidence interval width parameter to
\begin{equation}
\delta_k = \sqrt{\frac{ 3\,\omega_k \ln(2/\epcomp)}{\gamma n}},
\end{equation}
in order to have a similar completeness error $\epcomp \approx
10^{-12}$ across the different protocols.\footnote{In practice one
  would fix the soundness error of the protocol. However, because the
  soundness error is also dependent on the extraction phase we instead
  assume independence of rounds and fix the completeness error.}

We suppose that the devices operate
by using a pure, entangled state of the form
\begin{equation}
\ket{\psi(\theta)}_{AB} = \cos(\theta) \ket{00} + \sin(\theta) \ket{11},
\end{equation}
for $\theta \in (0,\pi/4]$. We denote the corresponding density
operator by $\rhotheta = \ketbra{\psi(\theta)}$. For simplicity we
restrict to projective measurements within the $x$-$z$ plane of the
Bloch-sphere, i.e., measurements $\lbrace \Pi(\varphi),\mathbb{1}-\Pi(\varphi)\rbrace$, with the projectors defined by
\begin{equation}
\Pi(\varphi) = \begin{pmatrix}
\cos^2(\varphi/2) & \cos(\varphi / 2) \sin(\varphi / 2) \\
\cos(\varphi / 2) \sin(\varphi/2) & \sin^2(\varphi/2)
\end{pmatrix}
\end{equation}
for $\varphi \in [0,2\pi)$. We denote the projectors associated with
the $j^{\text{th}}$ outcome of the $i^{\text{th}}$ measurement by
$A_{j|i}$ and $B_{j|i}$. The elements of the devices' behaviour can
then be written as
\begin{equation}\label{eq:Born-rule}
\pabgxy = \tr{\rhotheta(A_{a|x}\otimes B_{b|y})}.
\end{equation}

Our analysis is focussed on how the accumulation rates differ when the devices operate 
with inefficient detectors. Heralding can be used to
account for losses incurred during state transmission and has been
used to develop novel device-independent
protocols~\cite{MKSCBA18}. However, losses that occur within a user's
laboratory cannot be ignored without opening a
detection loophole~\cite{pearle}. Inefficient detectors are a major
contributor to the total experimental noise, so robustness to
inefficient detectors is a necessary property for any practical
randomness expansion protocol. We characterize detection efficiency by
a single parameter $\eta \in [0,1]$, representing the (independent)
probability with which a measurement device successfully measures a
received state and outputs the result.\footnote{For simplicity, we
  make the additional assumption that the detection efficiencies are constant
  amongst all measurement devices used within the protocol.} To deal
with failed measurements we assign outcome $0$ when this occurs. Combining this with~\eqref{eq:Born-rule}, we may write the behaviour as
\begin{equation}
\begin{aligned}
\pabgxy &= \eta^2\tr{\rhotheta(A_{a|x}\otimes B_{b|y})} + (1-\eta)^2 \delta_{0a}\delta_{0b} \\
&+ \eta(1-\eta)\left(\delta_{0a}\tr{\rhotheta(\identity\otimes B_{b|y})} +  \delta_{0b}\tr{\rhotheta(A_{a|x}\otimes \identity)}\right).
\end{aligned}
\end{equation}

For each protocol we consider lower bounds on two quantities: the
pre-EAT gain in min-entropy from a single interaction,
$\hmin(AB|XYE)$, and the \emph{EAT-rate},
$\hmin^{\epsmooth}(\AA\BB|\XX\YY\EE)/n$. The former quantity, which we
refer to as the \emph{asymptotic rate}, represents the maximum accumulation
rate achievable with our numerical technique. It is a lower bound on $\hmin^{\epsmooth}(\AA\BB|\XX\YY\EE)/n$, specified by \eqref{eq:EAT-total}, as $n \rightarrow \infty$ and $\gamma$, $\delta \rightarrow 0$.\footnote{We would really like to plot $H(AB|XYE)$ and the corresponding EAT-rate derived from it. However, in general we do not have suitable techniques to access these quantities in a device-independent manner.} Comparing these two quantities gives a clear picture of the amount of entropy that we lose due to the effect of finite statistics. 

With inefficient detectors, partially entangled states can exhibit
larger Bell-inequality violations than maximally entangled
states~\cite{Eberhard}. To account for this we optimize both the state
and measurement angles at each data point using the iterative
optimization procedure detailed in~\cite{ATL16}. All programs were relaxed to the second level of the NPA hierarchy using \cite{ncpol2sdpa} and the resulting SDPs were computed using the SDPA solver \cite{sdpa}. The results of these
numerics are displayed in \fig\ref{fig:detection-plots}.

\begin{figure}[t]
	\begin{center}
	\includegraphics[]{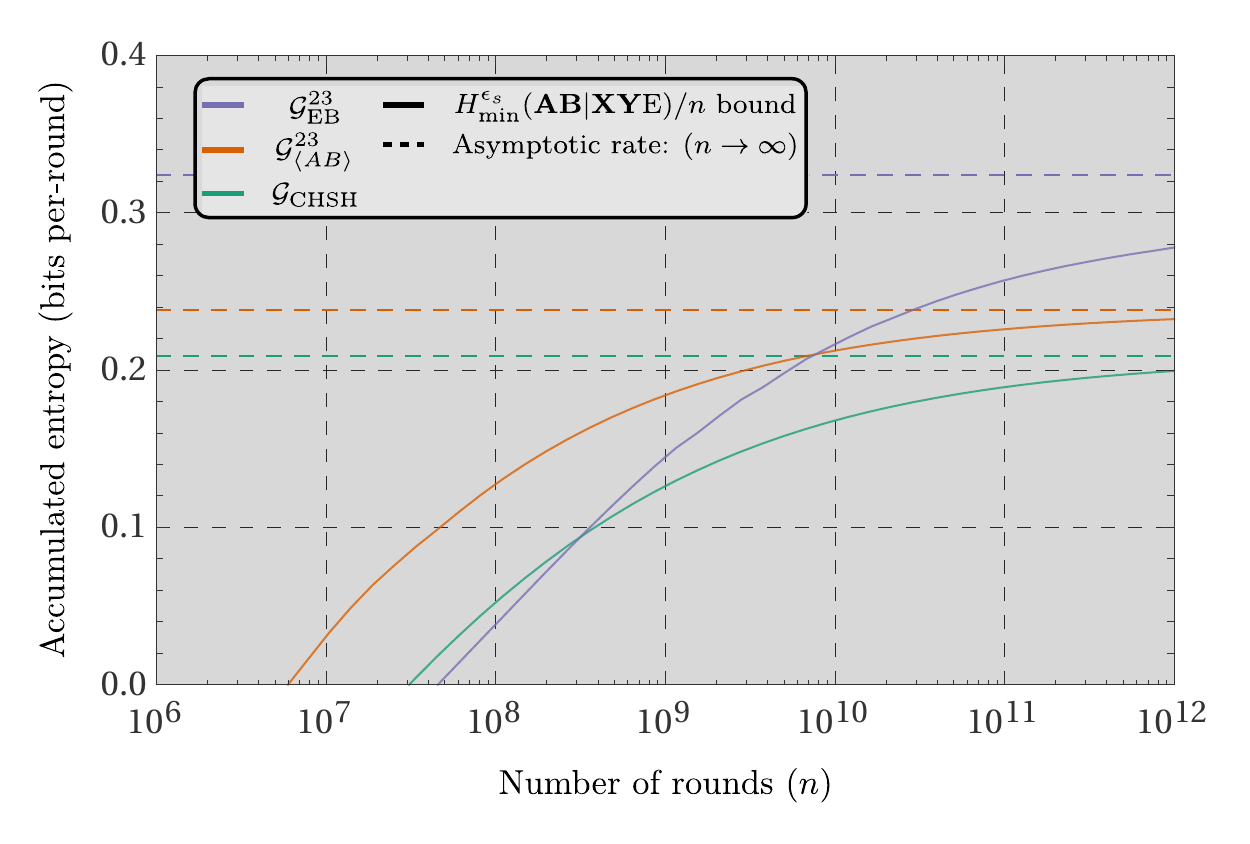}
	\end{center}\vspace{-0.5cm}
\caption{Comparison illustrating the EAT-rates (cf. \eqref{eq:EAT-total}) converging to the asymptotic rates for protocols based on different nonlocal games. The rates were derived by assuming a qubit implementation of the protocols with a detection efficiency $\eta = 0.9$, optimizing the state and measurement angles in order to maximise the asymptotic rate. Then, for each value of $n$ we optimized the min-tradeoff function choice and $\beta$ parameter and noted the resulting bound on $\hmin^{\epsmooth}$. To ensure that we approach the asymptotic rate as $n$ increased we set $\gamma = \delta_1 = \dots = \delta_{|\G|} = n^{-1/3}$, resulting in a constant completeness error across all values of $n$.}
\label{fig:RD-comparison}
\end{figure}

In \fig\ref{subfig:corrplot} and \fig\ref{subfig:ebplot} we see that in both families of protocols considered, an increase in the number of inputs leads to higher rates. This increase is significant when one moves from the $(2,2)$-scenario to the $(2,3)$-scenario. However, continuing this analysis for higher numbers of inputs we find that any further increases appear to have negligible impact on the overall robustness of the protocol.\footnote{This could also be an artefact of the assumed restriction to qubit systems.} Whilst all of the protocols achieve asymptotic rates of $2$ bits per round when $\eta = 1$, their respective EAT-rates at this point differ substantially. 
In \fig\ref{subfig:23comparison} we see a direct comparison between
protocols from the  different families. The plot shows that, as
expected, entropy loss is greater when using the nonlocality test
$G_{\FB}^{23}$ as opposed to the other protocols. In particular, for
high values of $\eta$ we find that we would be able to certify a
larger quantity of entropy by considering fewer scores. However, it is still worth noting that this entropy loss could be reduced by choosing a more generous set of protocol parameters, e.g., increasing $n$ and decreasing $\delta$.

Increasing $n$ can be difficult in practice due to restrictions on the overall runtime of the protocol. Not only does it take longer to collect the statistics within the device-interaction phase, but it may also increase the runtime of the extraction phase \cite{MPS12}. In \fig\ref{fig:RD-comparison}
we observe how quickly the various protocols converge on their
respective asymptotic rates as we increase $n$. Again we find that, due to the finite-size effect, entropy loss when using $\G^{23}_{\FB}$ is greater than that observed in the other protocols. In particular, we see that for protocols with fewer than $10^{10}$ rounds, it is advantageous to use $\G_{\al}^{23}$. From the perspective of practical implementation, \fig\ref{subfig:23comparison} and \fig\ref{fig:RD-comparison} highlight the benefits of a flexible protocol framework wherein a user can design protocols tailored to the scenario under consideration.

\begin{figure}
	\begin{center}
		\includegraphics[scale=1]{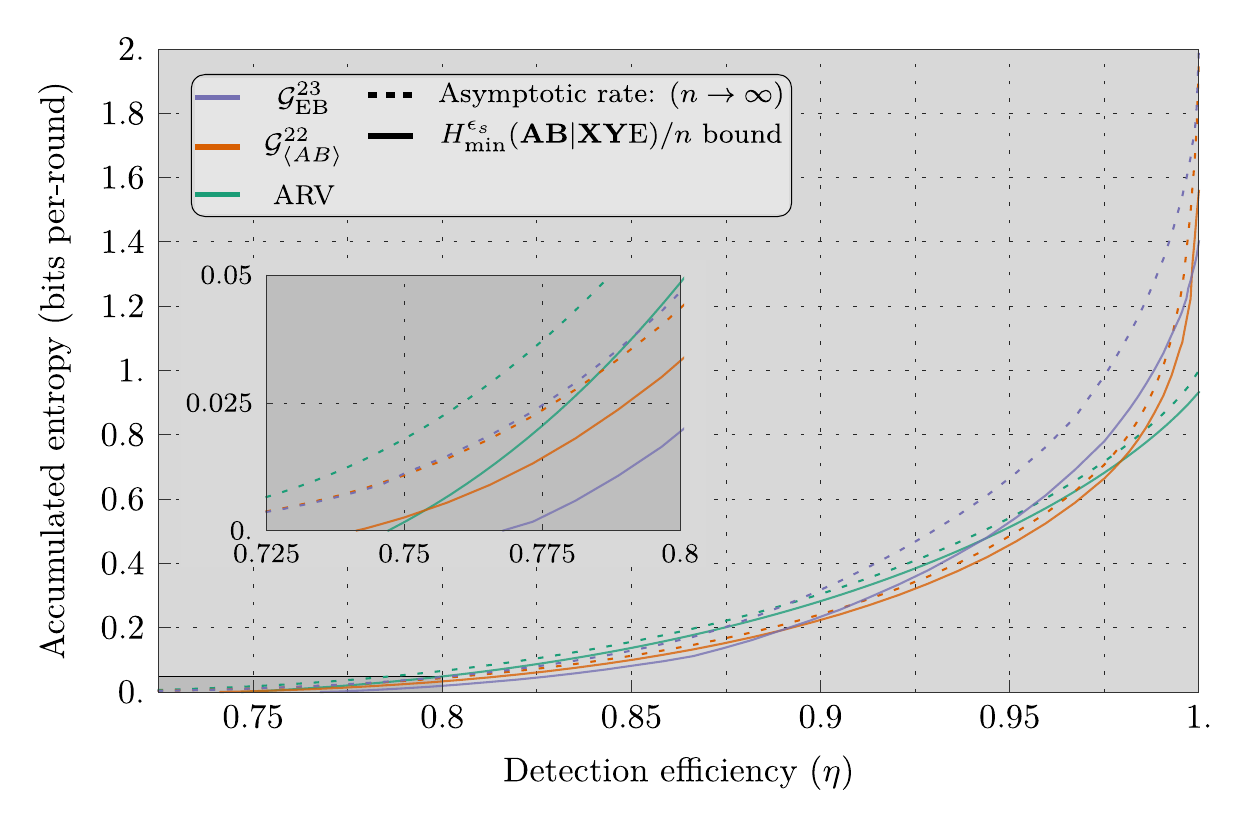}
	\end{center}
	\caption{Comparison between the certifiable accumulation rates
          of QRNE protocols based on $\G_{\chsh}$, $\G_{\FB}^{23}$ and
          Protocol ARV from~\cite{ARV} on qubit systems with
          inefficient detectors
          (cf. \fig\ref{fig:detection-plots}). The rates of Protocol
          ARV are also evaluated using the improved EAT
          statement~\cite{DF}. For Protocol ARV, we use the one-sided
          von Neumann entropy bound, so the maximum rate is one bit
          per round, but because we can directly get the single-round
          von Neumann entropy, the rate initially falls more slowly with
          decreasing detection efficiency than for the other
          protocols.} \label{fig:detection-AFRV-comparison}
	
\end{figure}

It is also important to compare the rates of instances of
Protocol~\protoName~with other protocols from the literature, in
particular the protocol of~\cite{ARV} (ARV). In~\cite{ARV}, the
min-tradeoff functions are constructed from a tight bound on the
single-party von Neumann entropy, $H(A|XE)$, which is given in terms
of a CHSH inequality violation~\cite{ABGMPS}. In
\fig\ref{fig:detection-AFRV-comparison} we compare the rates of ARV
with $\G_{\al}^{22}$ and $\G_{\FB}^{23}$ for entangled qubit systems
with inefficient detectors. To make our comparison fair, we have also
computed the rates for Protocol ARV using the improved EAT
bound\footnote{Note that we always use the direct bound on the von
  Neumann entropy when considering Protocol ARV, rather than forming a
bound via the min-entropy}. As the rates of Protocol ARV are derived
from the entropy accumulated by a single party their rates are capped
at one bit per round.

In contrast, the semidefinite programs grant us access to bounds on
the entropy produced by both parties and we are therefore able to
certify up to two bits per round. In
\fig\ref{fig:detection-AFRV-comparison}, this advantage is observed in
the high detection efficiency
regime. \fig\ref{fig:detection-AFRV-comparison} also highlights a
significant drawback of our technique, which stems from our use of the
inequality $H(AB|XYE) \geq H_{\min}(AB|XYE)$. In particular, we see
that for $\eta < 0.9$, the $H(A|XE)$ bound for the CHSH inequality is
already greater than the $H_{\min}(AB|XYE)$ established for the
empirical behaviour. Therefore, in the asymptotic limit
($n\rightarrow\infty$) the min-entropy bounds for these protocols will
produce strictly worse rates in this regime. For the finite $n$ we
have chosen, $n = 10^{10}$, it appears that for the majority of smaller $\eta$, it is advantageous to use the ARV protocol over the protocols derived from the framework. Nevertheless, looking at the threshold detection efficiencies, i.e. the minimal detection efficiency required to achieve positive rates, we find that some protocols from our framework are able to again beat the rates established for Protocol ARV. Looking at the inset plot in \fig\ref{fig:detection-AFRV-comparison} we see that $\G_{\al}^{22}$ has a smaller threshold efficiency than that of Protocol ARV for the chosen protocol parameters. Interestingly, this shows that $\G_{\al}^{22}$ is capable of producing higher rates than Protocol ARV in both the low and the high detection efficiency regimes, with the improvement for low detection efficiencies being of particular relevance to experimental implementations. Importantly, this shows that protocols from the framework are of practical use for finite $n$ in spite of the losses coming from the use of $H(AB|XYE) \geq H_{\min}(AB|XYE)$. 

\begin{remark}
We have so far considered the only noise to be that caused by
inefficient detectors. However, it is natural to ask how other sources
of noise affect our results. By replacing the states used with Werner
states~\cite{Werner}, we find that the results remain robust---they
remain qualitatively the same, but for small Werner state noise, all
of the graphs shift to slightly lower rates.  For this reason we
choose not to include the graphs here.
\end{remark}

\section{Conclusion}\label{sec:conclusion}
We have shown how to combine device-independent bounds on the guessing
probability with the EAT, to create a versatile method for analysing
quantum-secure randomness expansion protocols. The construction was
presented as a template protocol from which an exact protocol can be
specified by the user. The relevant security statements and quantity
of output randomness of the derived protocol can then be evaluated
numerically. A Python package~\cite{dirng-github} accompanies this
work to help facilitate implementation of the framework. In
\sec\ref{sec:examples} we illustrated the framework, applying it to
several example protocols, with parameters chosen to reflect the capabilities of current nonlocality tests. We then compared the robustness of these
protocols when implemented on qubit systems with inefficient
detectors. Our analyses show that, within a broadly similar
experimental setup, different protocols can have significantly
different rates, and hence that it is worth considering small
modifications to a protocol during their design. We also compared the rates of a selection of 
our protocols to the protocol presented in~\cite{ARV} (ARV). Interestingly, we found that some of the protocols from 
the framework are able to achieve higher rates than Protocol ARV in both the high and low detection efficiency regimes. In particular,
the higher rates for low detection efficiencies is of great importance for actual experimental implementations.

Although the framework produces secure and robust protocols, there
remains scope for further improvements. For example, our work relies
on the relation $H(AB|XYE) \geq \hmin(AB|XYE)$ which is far from
tight. The resulting loss can be seen when one compares the asymptotic
rate of $\G_{\chsh}$ in \fig\ref{subfig:23comparison} with those
presented in~\cite{ARV} (see
\fig\ref{fig:detection-AFRV-comparison}). Several alternative
approaches could be taken in order to reduce this loss. Firstly, the
above relation is part of a more general ordering of the conditional
R\'enyi entropies.\footnote{The R\'enyi entropies are one of many
  different entropic families that include the von Neumann entropy as
  a limiting case. Any such family could be used if they satisfy an
  equivalent relation.} If one were able to develop efficient
computational techniques for computing device-independent lower bounds
on one of these alternative quantities we would expect an immediate
improvement. Furthermore, dimension-dependent bounds may be applicable
in certain situations. For example, it is known that for the special
case of $n$-party, $2$-input, $2$-output scenarios it is sufficient to
restrict to qubit systems~\cite{Cirelson93,ABGMPS}.

Optimizing the choice of min-tradeoff function over $\Fmin$ is a
non-convex and not necessarily continuous
problem~\cite{curchod2017unbounded}. Our analysis in
\sec\ref{sec:examples} used a simple probabilistic gradient ascent
algorithm to approach this problem. We found that for certain
protocols, in particular $\G_{\FB}^{22}$, the optimization had to be
repeated many times before a good choice of min-tradeoff function was
found.

As \fig\ref{fig:detection-AFRV-comparison} shows, the framework is
capable of producing protocols that are of immediate relevance to
current randomness expansion experiments.  It is therefore a
worthwhile endeavour to search for protocols within the framework that
provide high EAT-rates in different parameter regimes. Investigations
into the randomness certification properties of nonlocality tests with
larger output alphabets or additional parties could be of
interest. However, increasing either of these parameters is likely to
increase the influence of finite-size effects. Alternatively, one
could try to design more economical nonlocality tests by combining
scores that are of a lesser importance to the task of certifying
randomness. Intuitively, for a score $c \in \mathcal{C}$, the
magnitude of of $\lambda(c)$ in the min-tradeoff function indicates
how important that score is for certifying entropy. If $|\lambda(c)|$
is large then this score is `important' in the sense that any small
deviations in the expected frequency of that score, $\omega(c)$, will
have a large impact on the amount of certifiable entropy.  Another
approach to designing good nonlocality tests would be to take
inspiration from \cite{NPS14,BSS14} wherein the authors showed how to
derive the optimal Bell-expressions for certifying randomness. A
nonlocal game could then be designed to encode the constraints imposed
by this optimal Bell-expression. An example of such a game would be to
assign a score $+1$ to all $(ABXY)$ that have a positive coefficient
in the optimal Bell-expression and a score of $-1$ to all those with
negative coefficients. The input distribution of the nonlocal game
could then be chosen as such to encode the relative weights of the
coefficients.

Finally, our computational approach to the EAT considered only the
task of randomness expansion. Our work could be extended to produce security proofs for other device-independent tasks. Given that the EAT has already been successfully applied to a wide range of problems~\cite{RMW18,RMW16,AK17-DIRA,AB17,BMP17}, developing good methods for robust min-tradeoff function constructions represents an important step towards practical device-independent security. 

\subsection*{Acknowledgments}
We are grateful for support from the EPSRC's Quantum Communications
Hub (grant number EP/M013472/1), an EPSRC First Grant (grant number
EP/P016588/1) and the WW Smith fund.

%\bibliography{dirng}{}
%\bibliographystyle{ieeetr}

\pagebreak
\appendix
\section{Table of parameters and notation}
\begin{center}
\begin{tabular}{|l|l|l|}
\hline
\textbf{Notation}   & \textbf{Description}  & \multicolumn{1}{l|}{\textbf{Initial reference}} \\
\hline
$\DD$ & A collection of untrusted devices. & Section~\ref{sec:devices-and-games} \\
\hline
$\G$  & A nonlocal game. & Definition~\ref{def:nonlocalgame} \\
\hline
$\Q_\G$  & Set of expected frequency distributions on $\G$ using quantum strategies.     & Section~\ref{sec:devices-and-games} \\
\hline
$\Q^{(k)}_\G$  & Set of expected frequency distributions on $\G$ using strategies from $\Qk$. & Section~\ref{sec:devices-and-games} \\
\hline
$\v, \w$  & Expected frequency distributions over scores of a nonlocal game. & Equation~\ref{eq:expected-freq-dist} \\
\hline
$\pk$, $\dk$  & Solutions to the $k$-relaxed primal and dual guessing probability programs.  & Program~\ref{prog:relaxed-primal} and \ref{prog:relaxed-dual} \\
\hline
$\lv$ & Feasible point of the dual guessing probability program with parameter $\v$. & Section~\ref{sec:devices-and-games} \\
\hline
$\d$  & Vector of statistical confidence interval widths. & Equation~\ref{eq:success-event}. \\
\hline
$\dpm$ & $\d$ with elements signed in accordance with a given $\l$. & Lemma~\ref{lem:accumulated-entropy} \\
\hline
$\A, \B$           & Devices' output alphabets.                                     & Section~\ref{sec:notation} \\ 
\hline
$\X, \Y$           & Devices' input alphabets.                                      & Section~\ref{sec:notation} \\ 
\hline
$n \in \mathbb N$     & Number of rounds in the device-interaction phase.                                          & Section~\ref{sec:accumulation-procedure} \\ 
\hline
$\gamma \in (0,1)$    & Probability that any given round is a test round.                    & Section~\ref{sec:accumulation-procedure} \\ 
\hline
$A_i, B_i$     & Devices' outputs for the $i^{\text{th}}$ round.                               & Section~\ref{sec:accumulation-procedure} \\ 
\hline
$X_i, Y_i$     & Devices' inputs for the $i^{\text{th}}$ round.                               & Section~\ref{sec:accumulation-procedure} \\ 
\hline
$\Ci$  & EAT-score for the $i^{\text{th}}$ round.                                            & Section~\ref{sec:accumulation-procedure} \\ 
\hline
$\freqcbm$  & Frequency distribution induced by score transcript $\CC=(C_1,\dots,C_n)$.                                             & Equation \ref{eq:frequency-distribution-scores} \\ 
\hline
$\Omega$ & Event that the protocol does not abort.
& Equation~\ref{eq:success-event} \\
\hline
%                                                              \\ \hline
$\ext$    & Strong quantum-secure randomness extractor. & Definition~\ref{def:extractor} \\
\hline
$\epcomp$  & Completeness error of Protocol~\protoName. & Lemma~\ref{lem:complet}                                   \\
\hline
$\epsound$ & Soundness error of Protocol~\protoName. & Lemma~\ref{lem:sound} \\
\hline
$\eps_s$  & Smoothing parameter for $H_{\min}$. & Equation~\ref{eq:smooth-entropy}                                                     \\
\hline
$\epeat$ & Tolerance of unlikely success events. & Lemma~\ref{lem:accumulated-entropy}.                                                     \\
\hline
$\errV$              & EAT error term (Variance).                                                     & Theorem~\ref{thm:EAT} \\
\hline
$\errK$              & EAT error term (Remainder).                                                    & Theorem~\ref{thm:EAT} \\
\hline
$\errW$              & EAT error term (Pass probability).                                                     & Theorem~\ref{thm:EAT} \\
\hline
$\epext$              & Extractor error.                                                     & Definition~\ref{def:extractor} \\
\hline
$\eploss$             & Entropy lost during extraction.                                      & Section~\ref{sec:extractors} \\
\hline
\end{tabular}%
\end{center}

\section{Input Randomness}
\label{sec:input-rand}
Here we quantify the length of the initial random seed required to execute an instance of Protocol~\protoName. This supply of random bits is necessary for selecting the devices' inputs and seeding the extractor. In the forthcoming analysis we ignore the latter as this quantity depends on the choice of extractor. Instead, we look at the process of converting a uniform private seed into device inputs required for running Protocol~\protoName. We follow a similar procedure to that used in~\cite{ARV}, modifying the algorithm slightly in order to extract explicit bounds.

\subsection{Statistical bounds}
We begin by stating some standard statistical bounds.  The first is
commonly known as the Chernoff bound~\cite{chernoff-bound}, although
we take our formulation from~\cite{HR90}.  This provides a convenient
bound on the deviation of the sum of random variables from the
expected value.
\begin{lemma}[Chernoff bound]\label{lem:Chernoff}
Let $X_i$ be independent binary random variables for $i=1,\ldots,n$, $S=\sum_i
X_i$ and $\mu=\E{S}$. Then for $0\leq t\leq1$
\begin{align*}
\pr{S\geq(1+t)\mu}&\leq\expo^{-t^2\mu/3}\\
\pr{S\leq(1-t)\mu}&\leq\expo^{-t^2\mu/2}\,.
\end{align*}
\end{lemma}
\begin{corollary}\label{cor:Chernoff}
For $r\leq\mu$ we have $\displaystyle\pr{\left|S-\mu\right|\geq r}\leq2\expo^{-r^2/(3\mu)}$.
\end{corollary}

In addition to this, we also make use of Hoeffding's inequality \cite{hoeffding-bound}.

\begin{lemma}[Hoeffding's inequality]
	Let $X_i$ be independent random variables, such that $a_i \leq X_i \leq b_i$ with $a_i, b_i \in \R$ for $i=1,\dots,n$. In addition, let $S=\sum_i
	X_i$ and $\mu=\E{S}$. Then for $t > 0$ 
	\begin{equation*}
	\pr{|S - \mu| \geq t} \leq 2 \expo^{-\frac{2t^2}{\sum_i (b_i - a_i)^2}}
	\end{equation*}
\end{lemma}

\subsection{Rounded interval algorithm}
The interval algorithm provides an efficient method for simulating the sampling of some target random variable $T$ using another random variable $S$. To aid understanding of our modification to this algorithm and any subsequent results we shall briefly explain how this simulation works. For simplicity we restrict ourselves to the scenario where $S$ is a sequence of uniformly distributed bits, we denote the uniform distribution on an alphabet of size $2^k$ by $U_{2^k}$, for $k \in \mathbb N$. 

The distribution of the target random variable $T$ forms a partition
of the unit interval, one subinterval for each outcome $t$ of $T$. In
exactly the same way, the probability distribution for $U_{2^k}$
partitions the unit interval into $2^k$ subintervals. Thus, we can
associate a bit-string with its corresponding subinterval, defined by
this partitioning. The interval algorithm works by using an increasing sequence of random bits and the corresponding subintervals that the sequence defines. Once the subinterval generated by the sequence of bits is contained inside one of the subintervals defined by the target random variable $T$, say $t$, then we say that we have simulated the sampling of $t$ from $T$ and the algorithm terminates. Denoting by $N$ the length of seed required for the interval algorithm to terminate, then by~\cite[Theorem~3]{H97} we have
	\begin{equation}\label{eq:expected-seed-length}
	\E{N} \leq H(T) + 3.
	\end{equation}

As the algorithm stands, the maximum value that $N$ can take is
unbounded (although the probability that the algorithm expends the
seed without terminating decreases exponentially in $N$). In order to produce large deviation bounds on the number of bits required to execute our protocol, we place an upper limit on the maximum seed length. We thus use an adapted sampling procedure, the \textit{rounded interval algorithm} (RIA), which forcefully terminates if the seed length reaches the upper bound of $k_{\max}$ bits.  

Should the RIA fail to terminate after $\kmax$ steps, then the output sequence generated will correspond to some subinterval $I(r) = [\frac{r}{2^{\kmax}}, \frac{r+1}{2^{\kmax}})$, for some $r\in\{0,1,\dots,2^{\kmax}-1\}$, that is not entirely contained within one of the subintervals induced by $T$. If this occurs, we \emph{round down}: selecting the interval $I_t$ for which $\frac{r}{2^{\kmax}} \in I_t$.

\begin{remark}
Note that the above procedure depends on the ordering of the intervals generated by $T$ (which should be fixed before sampling). One could imagine a rather pathological scenario where an ordering places extremely unlikely outcomes over rounding points, greatly increasing their simulated outcome probabilities. However, as will be shown in \lem\ref{lem:ria-distance}, the distance between the simulated random variable and the target random variable decreases exponentially in $\kmax$.
\end{remark}

\begin{remark}\label{rem:joint-sampling-efficiency}
The rounding procedure truncates the maximum seed length, $N \leq \kmax$, and as such, it is clear that the inequality \eqref{eq:expected-seed-length} also holds for the RIA.  
\end{remark}

\begin{definition}[Statistical distance] Given two random variables $X$ and $X'$, taking values in some common alphabet $\mathcal{X}$. The \emph{statistical distance} between $X$ and $X'$, is defined by 
\begin{equation}
\Delta(X,X') := \frac12 \sum_{x\in \mathcal{X}} |p_{X}(x) - p_{X'}(x)|.
\end{equation} 
\end{definition}

\begin{lemma}\label{lem:ria-distance}
Let $T$ be a random variable taking values in some alphabet $\mathcal{T}$. Let $T'$ be the distribution sampled using the RIA with target distribution $T$. Then 
$$
\Delta(T,T')\leq |\mathcal{T}|\, 2^{-(\kmax+1)},
$$
where $k_{\max}$ is the maximum number of input bits that can be used by the RIA.
\begin{proof}
Consider the partitions of the unit interval $\{I(t)\}_{t \in \mathcal{T}}$ and $\{I'(t)\}_{t \in \mathcal{T}}$ corresponding to the distributions $p_T$ and $q_{T'}$ of $T$ and $T'$ respectively. The intervals of $T'$ take the form
$$
I'(t) = \bigcup_r \left[\frac{r}{2^{\kmax}}, \frac{r+1}{2^{\kmax}}\right)
$$
where the (potentially empty) union is taken over all $r \in \mathbb{N}_0$ such that $r 2^{-\kmax} \in I(t)$. The intervals within the union are either contained fully within the corresponding outcome interval of $T$, i.e., $\left[\frac{r}{2^{\kmax}}, \frac{r+1}{2^{\kmax}}\right) \subseteq I(t)$, or they are included as a result of rounding. Thus we may write
$$
|I'(t)| = |\{r \mid r\cdot 2^{-\kmax} \in I(t), r \in \mathbb{N}_0\}|2^{-\kmax}.
$$
By a straightforward counting argument, there are at least $\floor{|I(t)|2^{\kmax}}$ such values of $r$, and at most $\ceil{|I(t)|2^{\kmax}}$. We hence have 
$$
|I(t)|2^{\kmax} - 1 \leq |I'(t)|2^{\kmax} \leq |I(t)|2^{\kmax} + 1,
$$ 
and therefore
$$
|p_T(t) -p_{T'}(t)| \leq 2^{-\kmax},
$$
holds for all $t \in \mathcal{T}$. Applying this bound to each term within the $\Delta(T,T')$ sum completes the proof.
\end{proof}
\end{lemma}

\subsection{Input randomness for Protocol~\protoName}

Following the structure of Protocol~\protoName, we look to use the RIA
to sample the devices' inputs for each round. In adherence with the
Markov-chain condition (\defn\ref{def:eat-channels}), the natural
procedure would be to sample at the beginning of each round. However,
in practice this requires a much larger seed: because
of~\eqref{eq:expected-seed-length} and the property $H(T^n)=nH(T)$, by sampling the joint distribution 
the expected saving is about $3n$ bits compared to repeating a single
sample $n$ times. Fortunately, this joint sampling can be implemented
while maintaining the Markov-chain condition. Within the assumptions
of Protocol~\protoName\ we allow the honest parties access to a
trusted classical computer, which would also contain some trusted data
storage---we assume that the parties can record their outcome strings
without leakage. Therefore, using their trusted classical computer,
the honest parties perform the RIA: sampling the random variables
$(X_1^n,Y_1^n)$ and subsequently storing the outcome on the trusted
classical computer's harddrive. Crucially, the assumption that the
user can prevent unwanted communication between devices implies this can all be done without any information leaking to the untrusted devices. Then, at the beginning of round $i$, the inputs $(X_i, Y_i)$ are sent from the classical computer to the respective devices. By conducting the protocol in this manner we retain the Markov chain conditions---the inputs are sampled independently of the devices and furthermore, when the devices produce their outputs for the $i^{\text{th}}$ round they can only have knowledge of the inputs for this round and all previous.

Due to potential computational constraints and to permit large deviation bounds on the number of bits required we will not assume that all $n$ rounds are sampled at once. Instead, we split the $n$ rounds into at most $\ceil{n/m}$ blocks of size $m$ and apply the RIA to sample the inputs of each block separately. For simplicity, we assume that $n/m \in \mathbb{N}$ and henceforth remove the ceiling function from the analysis. 

Recall that for the $i^{\text{th}}$ round, the user first uses $T_i$
to decide whether the round is a test round, and, if so, they choose
inputs according to the nonlocal game input distribution $\mu$. Otherwise, if $T_i = 0$, they supply their devices with the fixed inputs $\x$ and $\y$. The probability mass function of joint random variables $X_iY_iT_i$, representing the $i^{\text{th}}$ round's inputs, is therefore 
\begin{equation}
\pr{(X_i,Y_i,T_i) = (x_i,y_i,t_i)} =
\begin{aligned}
\begin{cases} 
\gamma\,\mu(x,y) &\text{for } (x_i,y_i,t_i) = (x,y,1), \\
(1-\gamma) &\text{for } (x_i,y_i,t_i) = (\x,\y,0) \\
\qquad 0 &\text{otherwise}
\end{cases}.
\end{aligned}
\end{equation}
Following \eqref{eq:expected-seed-length}, if $M$ is the seed length required to sample one of the $m$ blocks of rounds, then we have 
\begin{equation}\label{eq:expected-seed-size-block}
\E{M} \leq \frac{ \left(\gamma H(\mu) + h(\gamma)\right) n}{m} + 3
\end{equation}
where $H(\mu)$ is the Shannon entropy of the distribution $\mu$ and $h(\cdot)$ is the binary entropy.

The following lemma gives a probabilistic bound on the total length of the random seed required to sample the inputs for the devices.\footnote{We do not include the extractor's seed here as its size will depend on the choice of extractor.}

\begin{lemma}\label{lem:input-randomness}
Let the parameters of Protocol~\protoName\ be as defined in
\fig\ref{fig:full-protocol} and let $\kmax \in \mathbb N$ be the
maximum permitted seed length for an instance of the RIA. Then, with
probability greater than $(1-\epsamp)$, we can use $m$ instances of
the RIA to simulate the sampling of every device input required to
execute Protocol~\protoName\ with a uniform seed of length no greater than $\nseed$, where
\begin{align}\label{eq:samp}
\nseed &= 2 \kappa \\
\epsamp &= e^{-2\kappa^2/m\kmax^2}
\intertext{and $\kappa = \left(\gamma H(\mu) + h(\gamma)\right) n + 3 m$. Moreover, the sampled distribution lies within a statistical distance of}
\epdist &= m\,2^{n \log(\supp(\mu) + 1)/m -(\kmax+1)},
\end{align}
from the target distribution, where $\supp(\mu) := |\{(x,y) \in \X\Y \mid \mu(x,y) > 0\}|$.
\end{lemma}
\begin{proof} 
Consider the sequence $(M_i)_{i=1}^m$ of \iid random variables representing the number of random bits required to choose the inputs for the $i^{\text{th}}$ block and the corresponding random sum $N = \sum_{i=1}^m M_i$. By \eqref{eq:expected-seed-size-block}, the expected number of bits required to select all of the inputs for the protocol can be bounded above by $\kappa = \left(\gamma H(\mu) + h(\gamma)\right) n + 3 m$. Using Hoeffding's inequality, we can bound the probability that $N$ greatly exceeds this value,
\begin{equation*}
\pr{N \geq \kappa + t} \leq e^{-2t^2/m\kmax^2},
\end{equation*}
for some $t>0$. Setting $t=\kappa$ this becomes
\begin{equation*}
\pr{N \geq 2\kappa} \leq e^{-2\kappa^2 / m\kmax^2}.
\end{equation*}
Although $\kappa$ is not exactly the expected value of $N$, which is the quantity appearing in Hoeffding's bound, the bound holds because $\kappa\geq\E{N}$.

It remains to bound the statistical distance between the sampled random variable $\II' = (\XX',\YY',\TT')$ and the target random variable $\II = (\XX,\YY,\TT)$. For each block of rounds, the corresponding random variable $\II_i$ can take one of a possible $(\supp (\mu) + 1)^{n/m}$ different values. Therefore, by Lemma~\ref{lem:ria-distance}, we have for the $i^{\text{th}}$ block of rounds
\begin{align*}
\Delta(\II_i,\II_i') &\leq (\supp (\mu) + 1)^{n/m} 2^{-(\kmax+1)} \\
&= 2^{n \log(\supp(\mu) + 1)/m -(\kmax+1)}
\end{align*}
Since $\Delta(W,V)$ is a metric and hence satisfies the triangle
inequality~\cite{Cover&Thomas}, the statistical distance between
independently repeated samples can grow no faster than linearly, i.e.,
$\Delta(I^{m},I'^m) \leq m \Delta(I,I')$.  This completes the proof.
\end{proof}

\section{Incorporating the blocking procedure of \cite{ARV}}\label{app:EAT}
The original statement of the entropy accumulation theorem~\cite{DFR} was released alongside an accompanying paper,~\cite{ARV}, which detailed its application to security proofs of device-independent protocols. Within the appendix of~\cite{ARV} it was shown that one could increase the quantity of entropy certified by the original EAT by modifying the structure of the protocol. In particular, this demonstrated the original EAT statement's suboptimal dependence on the testing probability. In light of this, the authors of~\cite{DF} improved the second order term of the EAT in order to account for this suboptimal dependence. In the sections that follow, we will look at how the modified protocol structure interacts with the improved EAT statement. We begin by showing how the family of min-tradeoff functions $\Fmin$ can be adapted to this structural change and then we show that the modified protocol structure provides no clear benefits when used with the improved EAT statement. In addition, we provide some comparison plots showing the accumulation rates achievable with the different structures and EAT statements. To clearly distinguish the different statements of the EAT, we shall indicate with the subscript $~_{\DFR}$, quantities associated with the original EAT~\cite{DFR} and similarly we shall indicate with the subscript $~_{\DF}$, quantities associated with the EAT with improved second-order~\cite{DF}.

\subsection{Construction}
Let us briefly review the structural modification that was introduced in~\cite{ARV}. Instead of distinguishing the statistics from each interaction separately, rounds are grouped together to form \emph{blocks}. The number of rounds within a block can vary: a new block begins when either a test-round occurs or when the maximum number of rounds permitted within a block, $\smax$, is reached. On expectation there are $\sbar = \frac{1-(1-\gamma)^{\smax}}{\gamma}$ rounds within a block. The device-interaction phase of the protocol concludes after some specified number of blocks $m \in \mathbb{N}$ have terminated. We shall use the superscripts $~^{R}$ and $~^{B}$ to indicate whether a quantity is concerned with the round-by-round or block structured protocols respectively. 

The collected information is now defined at the level of blocks and not rounds. In particular, at the end of the $i^{\text{th}}$ block the user records some tuple $(\AA_i,\BB_i,\XX_i,\YY_i,\Ci)$, where $(\AA_i,\BB_i,\XX_i,\YY_i) \in \A^{\smax}\B^{\smax}\X^{\smax}\Y^{\smax}$ and the score's alphabet remains the same, $\Ci \in \G \cup \{\perp\}$. The EAT-channels are also defined for each block and the entropy bounding property (cf.~\eqref{eq:fmin-definition}) that the min-tradeoff functions must satisfy becomes
\begin{equation}\label{eq:blocked-entropy-bounding-prop}
f^B_{\min}(\p) \leq \inf_{\sigma_{R_{i-1}R'}: \N_i(\sigma)_{\CC_i}=\tau_{\p}} H(\AA_i\BB_i|\XX_i\YY_iR')_{\N_i(\sigma)},
\end{equation} 
for each $i \in [m]$. The set of distributions compatible with the protocol structure (cf. \eqref{eq:protocol-respecting-distribution-extended}) now take the form
\begin{equation}
\p^{B} = \begin{pmatrix}
\gamma \sbar \bm q \\
(1-\gamma)^{\smax}
\end{pmatrix}
\end{equation}
for $\bm q \in \Q_\G$.

\begin{lemma}[Blocked variant of Lemma~\ref{lem:extensionlemma}]\label{lem:blockedextensionlemma}
Let $g: \P_{\G} \rightarrow \R$ be an affine function satisfying 
\begin{equation}\label{eq:blocked-entropy-bounding-OVG}
g(\bm q) \leq \inf_{\sigma_{R_{i-1}R'}: \N^{\test}_i(\sigma)_{\Ci}=\tau_{\p}} H(\AA_i\BB_i|\XX_i\YY_iR')_{\N_i(\sigma)}
\end{equation}
for all $\bm q \in \Q_{\G}$. Then the function $f: \P_{\G\cup\{\perp\}}\rightarrow \R$, defined by its action on trivial distributions
\begin{align*}
f(\e_c) &= \Max{g} + \frac{g(\e_c) - \Max{g}}{\gamma\sbar}, \qquad \forall c \in \G, \\
f(\e_\perp) &= \Max{g},
\end{align*} 
is a min-tradeoff function for any EAT-channels implementing Protocol~\protoName$^B$. Furthermore, $f$ satisfies the following properties:
\begin{align*}
\Max{f} &= \Max{g}, \\
\Min{\frestricted} &\geq \Min{g}, \\
\Var{\frestricted} &\leq \frac{(\Max{g} - \Min{g})^2}{\gamma\sbar}.
\end{align*}
\begin{proof}
This follows from replicating the original proof~\cite{DF} with the block channels decomposed into the testing and generation channels, $\N_i = \gamma\sbar \N_i^{\test} + (1-\gamma\sbar)\N_i^{\gen}$.
\end{proof}
\end{lemma}

\begin{lemma}[Blocked min-tradeoff construction]\label{lem:blocked-fmin}
Let $\G$ be a nonlocal game and $k \in \mathbb N$. For each $\v\in\Qk_\G$, let $\lv$ be some feasible point of \prog\eqref{prog:relaxed-dual}. Furthermore, let $\lmax =
\max_{c \in \G}\lambda_{\v}(c)$ and $\lmin =
\min_{c \in \G}\lambda_{\v}(c)$. Then, for any set of EAT channels $\{\N_i\}_{i=1}^m$ implementing an instance of Protocol~\protoName$^B$ with the nonlocal game $\G$, the set of functionals $F^B_{\min}(\G) = \{f_{\v}(\cdot)\mid \v\in \Qk_\G\}$ forms a family of min-tradeoff functions, where $f_{\v}: \P_\C \rightarrow \R$ are defined by their actions on trivial distributions
\begin{align}\label{eq:blocked-fmin-c}
\hspace{4cm}f_{\v}(\e_{c}) &:= (1-\gamma) \,\sbar \left(\Av - \Bv \frac{\lv\cdot\e_c - (1-\gamma\sbar)\lmin}{\gamma\sbar} \right) \hspace{0.5cm}\text{for } c \in \G, \intertext{and}
f_{\v}(\e_\perp) &:= (1-\gamma)\,\sbar \left(\Av - \Bv\, \lmin \right), &~
\end{align}
where $\Av = \frac{1}{\ln 2} - \log (\lv\cdot\v)$ and $\Bv = \frac{1}{\lv\cdot\v \ln 2}$.

Moreover, these min-tradeoff functions satisfy the following identities.
\begin{itemize}
\item Maximum:
\begin{equation}\label{eq:blocked-fv-max}
\Max{f_{\v}} = (1-\gamma)\sbar (\Av - \Bv \,\lmin)
\end{equation}
\item $\respecting$-Minimum:
\begin{equation}\label{eq:blocked-fv-min}
\Min{\left.f_{\v}\right|_{\respecting}} \geq (1-\gamma)\sbar(\Av - \Bv \,\lmax)
\end{equation}
\item $\respecting$-Variance:
\begin{equation}\label{eq:blocked-fv-var}
\Var{\left.f_{\v}\right|_{\respecting}} \leq \frac{(1-\gamma)^2 \sbar \Bv^2 (\lmax - \lmin)^2}{\gamma}
\end{equation}
\end{itemize}
\begin{proof}
The proof follows the same structure as the proof of \lem\ref{lem:fmin}. The only significant difference is the construction of the function $g:\P_{\G}\rightarrow \R$ satisfying \eqref{eq:blocked-entropy-bounding-OVG} so we shall explain this part here. Following Appendix B
of~\cite{ARV}, by repeated application of the chain rule we may decompose a block's entropy as
\begin{align*}
H(\AA_i\BB_i|\XX_i\YY_i\TT_iR')_{\mathcal N_i(\sigma)} 
=\sum_{j=1}^{\smax} (1-\gamma)^{j-1} H(A_{i,j}B_{i,j}|\XX_i\YY_i,\TT_{i,1}^{~\,j-1}\!\!= \bm{0},\TT_{i,j}^{\,\,\,\smax} \AA_{i,1}^{~\,j-1}\BB_{i,1}^{~\,j-1}R'),
\end{align*}
where $T_{i,j}$ is the random variable indicating whether a test occurred on the $j^{\text{th}}$ round of the $i^{\text{th}}$ block.
Considering the individual terms within the sum, we can absorb the majority of the side information into some arbitrary quantum register $\EE$ leaving us with terms of the form
$$
(1-\gamma)^{j-1} H(A_{i,j}B_{i,j}|X_{i,j}Y_{i,j}T_{i,j}\EE).
$$
As before, we can use the inequality $H(A|B) \geq H_{\min}(A|B)$ and conditioning on $T_{i,j}$ to lower bound each term in the sum by the outputs of the semidefinite program, 
\begin{align*}
(1-\gamma)^{j-1} H(A_{i,j}B_{i,j}|X_{i,j}Y_{i,j}T_{i,j}\EE)
&=
(1-\gamma)^{j-1} \pr{T_{i,j}=0} H(A_{i,j}B_{i,j}|X_{i,j}=\x,Y_{i,j}=\y,T_{i,j} = 0,\EE) \\
&+
(1-\gamma)^{j-1} \pr{T_{i,j}=1} H(A_{i,j}B_{i,j}|X_{i,j}Y_{i,j}T_{i,j} = 1\,\,\EE) \\
&\geq 
(1-\gamma)^j H(A_{i,j}B_{i,j}|\x\,\y\,\EE) \\
&\geq
(1-\gamma)^j H_{\min}(A_{i,j}B_{i,j}|\x\,\y\,\EE) \\
&\geq 
-(1-\gamma)^j \log(\lv \cdot \w_{i,j}),
\end{align*}
where $\w_{i,j}\in\Q_{\G}$ is the expected frequency distribution over the games scores for round $j$ of block $i$. Noting that $-\log(\cdot)$ of a linear function is convex, we can establish a bound on the entire block $i$ through an application of Jensen's inequality
\begin{align*}
(\gamma-1)\sum_{j=1}^{\smax} (1-\gamma)^{j-1} \log(\lv\cdot \w_{i,j}) 
&\geq
\sbar (\gamma-1) \log\left(\lv\cdot \frac{\sum_{j=1}^{\smax} (1-\gamma)^{j-1} \w_{i,j}}{\sbar}\right) \\
&= 
\sbar (\gamma-1) \log\left(\lv\cdot \w_i\right) ,
\end{align*}
we have used the fact that $\sum_{j \in [\smax]} (1-\gamma)^{j-1} =
\sbar$ and that $ \w_i = \frac{\sum_{j \in [\smax]} \gamma
(1-\gamma)^{j-1} \w_{i,j}}{\gamma \sbar}$ is the normalised expected frequency distribution over the nonlocal game scores for the
$i^{\text{th}}$ block. Taking a first-order expansion of the last
line, we get the function $g_{\v}(\,\cdot\,) = (1-\gamma)\sbar\left(\Av -
  \Bv \lv \cdot (\,\cdot\,)\right)$. The proof is then completed by applying the extension
\lem\ref{lem:blockedextensionlemma}, analogous to the technique of \lem\ref{lem:fmin}.
\end{proof}
\end{lemma}
\subsection{Blocking with the improved second order}

The error term in the original EAT bound is
\begin{equation}
\errDFR^R := 2 \left( \log(1 + 2 |\A||\B|) + \ceil{\|\nabla f_{\min}\|_{\infty}}\right)\sqrt{1-2\log(\epsmooth\epeat)}.
\end{equation} 
The disadvantage of using this bound as-is is that the gradient of $f_{\min}$ scales like $1/\gamma$ and so the total error scales as $O(\sqrt{n}/\gamma)$. Collating the statistics into $m\in \mathbb{N}$ blocks, allows some of the $\gamma$ dependence from the gradient term to be transferred to the $\log(1+2|\A||\B|)$ term. Moving to the blocked structure and setting $\smax=\ceil{1/\gamma}$ (as was done in~\cite{ARV}), the output alphabets grow exponentially with the size of the block and the logarithmic term acquires a $1/\gamma$ scaling. In contrast, the scaling of the derivative of the min-tradeoff function is found to be independent of the block size. Fortunately, as our error is defined for an entire block, we reduce the multiplicative factor on the total error from $\sqrt{n}$ to $\sqrt{m}\approx\sqrt{n/\sbar}$. As $\sbar \in O(1/\gamma)$, we find that the total error term now scales as $\sqrt{n/\gamma}$. By increasing the size of the blocks we have effectively redistributed the testing probability dependence evenly amongst the components of $\errDFR$.

In~\cite{DF}, the authors looked to amend this deficiency by strengthening the second order term in the EAT. The following short calculation looks at how the errors scale when we applying the blocking procedure to the improved EAT statement. Recall the error terms
\begin{equation}
\errV^R := \frac{\beta \ln 2}{2} \left(\log\left(2 |\A\B|^2 +1\right) + \sqrt{  \Var{\frestricted} + 2} \right)^2,
\end{equation}
\begin{equation}
\errK^R := \frac{\beta^2}{6(1-\beta)^3\ln 2}\,
2^{\beta(\log |\A\B| + \Max{f} - \Min{\frestricted})}
\ln^3\left(2^{\log |\A\B| + \Max{f} - \Min{\frestricted}} + e^2\right)
\end{equation}
and 
\begin{equation}
\errW^R := \frac{1}{\beta}\left(1 - 2 \log(p_{\Omega}\,\epsmooth)\right).
\end{equation}
Using the explicit form of the blocked min-tradeoff functions \lem\ref{lem:blocked-fmin}, we can calculate the asymptotic growth of the error terms as $\smax \rightarrow \infty$, $\gamma \rightarrow 0$ and $m \approx n^R / \sbar$. In particular, we find 
\begin{equation}
\begin{aligned}
m \cdot \errV^B &\leq \frac{\beta m \ln 2}{2} \left(\log\left(2 |\A\B|^{2\smax} +1\right) + \sqrt{  \frac{(1-\gamma)^2 \sbar \Bv^2 (\lmax - \lmin)^2}{\gamma} + 2} \right)^2 \\
&= O(\beta n \smax) + O(\beta n/\gamma), 
\end{aligned}
\end{equation}
\begin{equation}
\begin{aligned}
m \cdot \errK^B &\leq \frac{m \beta^2}{6(1-\beta)^3\ln 2}\,
2^{\beta( \log |\A\B|^{\smax} + (1-\gamma) \sbar \Bv (\lmax - \lmin))}
\ln^3\left(2^{ \log |\A\B|^{\smax} + (1-\gamma)\sbar\Bv(\lmax - \lmin)} + e^2\right) \\
&= \beta^2 2^{O(\beta\smax)} O(n \smax^2), 
\end{aligned}
\end{equation}
\begin{equation}
\errW^B = O(1/\beta),
\end{equation}
and therefore the total error scales as
\begin{equation}
\epsilon^B_{\DF} = O\left(\beta n \smax + \frac{\beta n}{\gamma} +  \beta^2 n \smax^2 2^{O(\beta\smax)} + \frac{1}{\beta}  \right).
\end{equation}

In order for $\errK^B$ to have any sensible scaling, we need the
exponent to grow no faster than $O(1)$. Combining this with the
inverse dependence of $\beta$ in $\errW^B$, we would like $\beta\approx\frac{\sqrt{\gamma}}{\sqrt{n} \smax}$. Such a choice results in $\epsilon^B_{\DF} \in O\left(\smax \sqrt{n/\gamma}\right)$ which suggests that indeed, the blocking procedure is no longer advantageous when used in conjunction with the improved second order statement. 

A comparison between the expansion rates obtained when using the improved second order statement~\cite{DF} and the blocked variant of the original EAT are presented in \fig\ref{fig:EAT-RD-comparison}. The faster convergence to the asymptotic rate is indicative of the new EAT statement's strength.

\begin{figure}
\begin{center}
\begin{minipage}{0.45\textwidth}
\includegraphics[width=\textwidth]{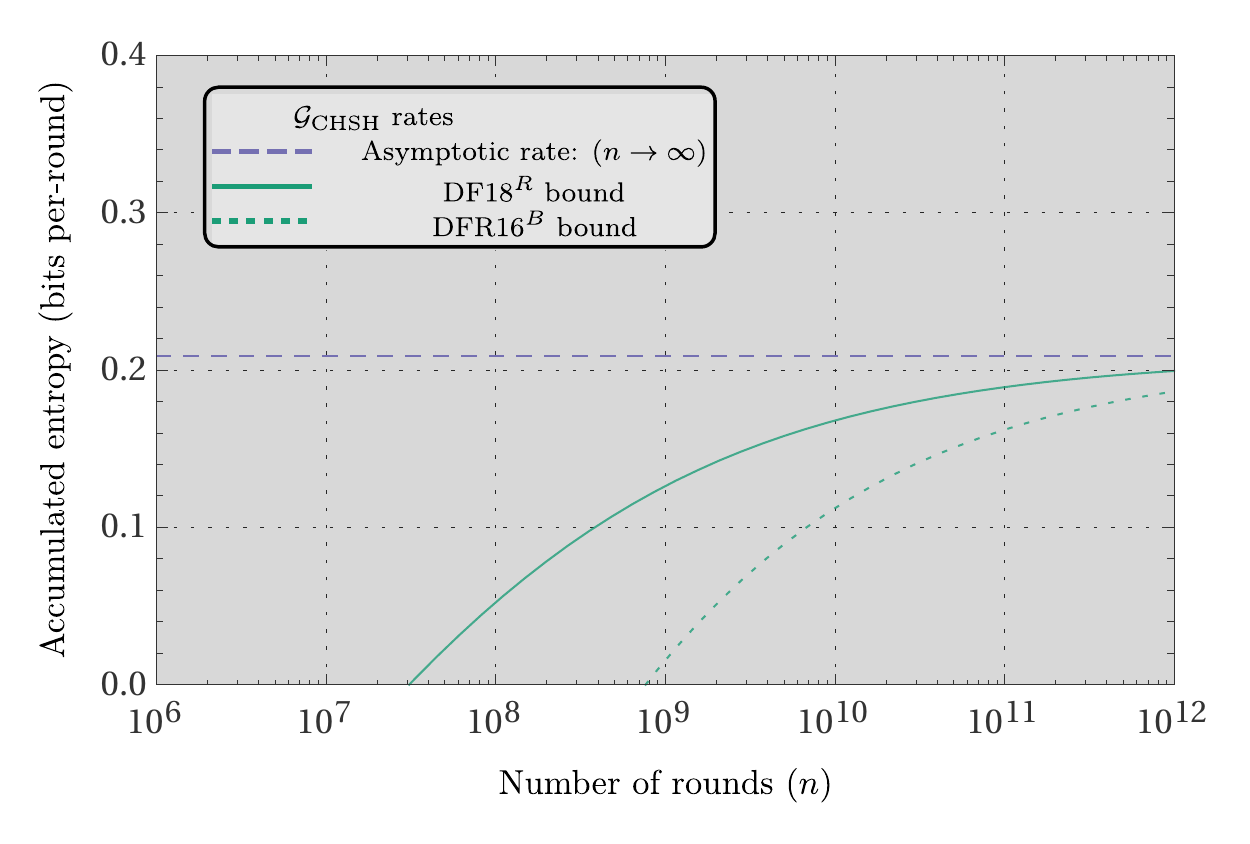}
\end{minipage}
\begin{minipage}{0.45\textwidth}
\includegraphics[width=\textwidth]{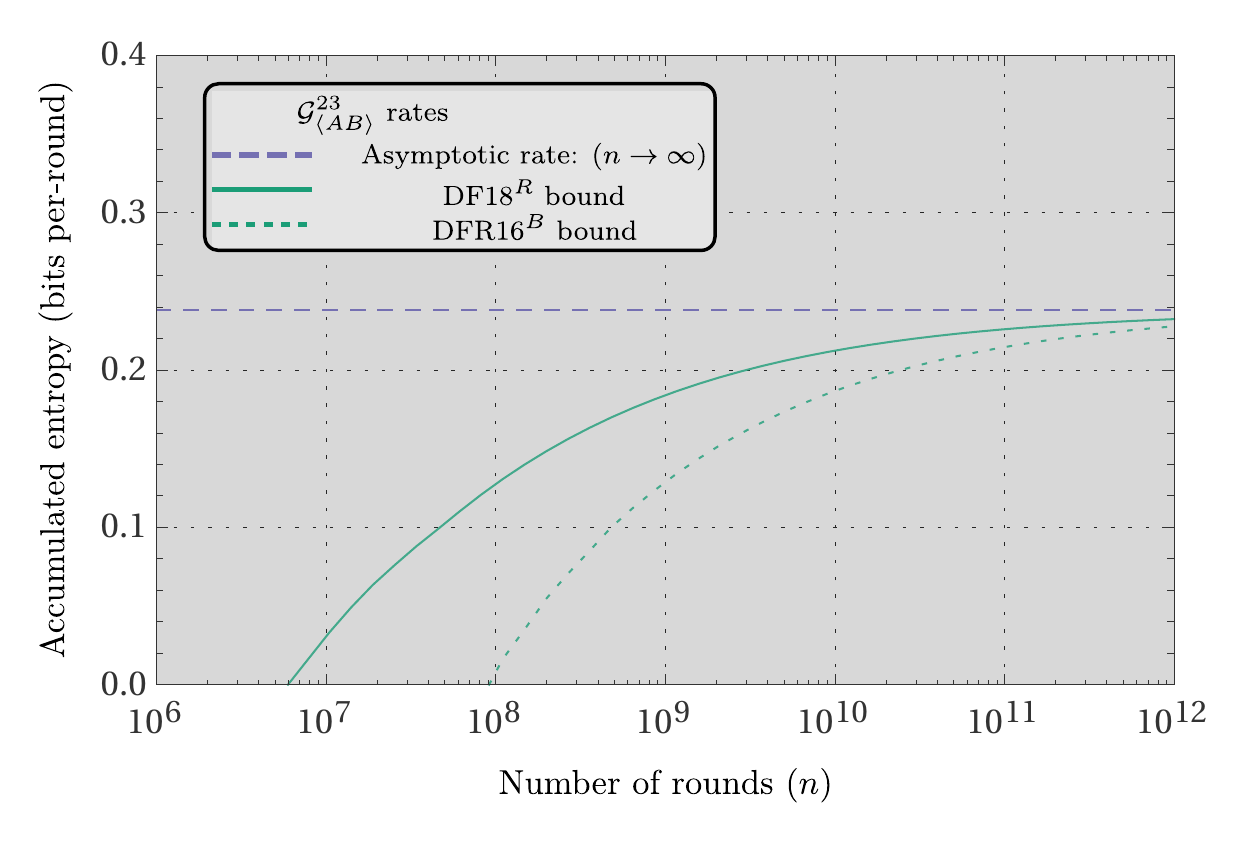}
\end{minipage} \\
	\begin{minipage}{0.45\textwidth}
		\includegraphics[width=\textwidth]{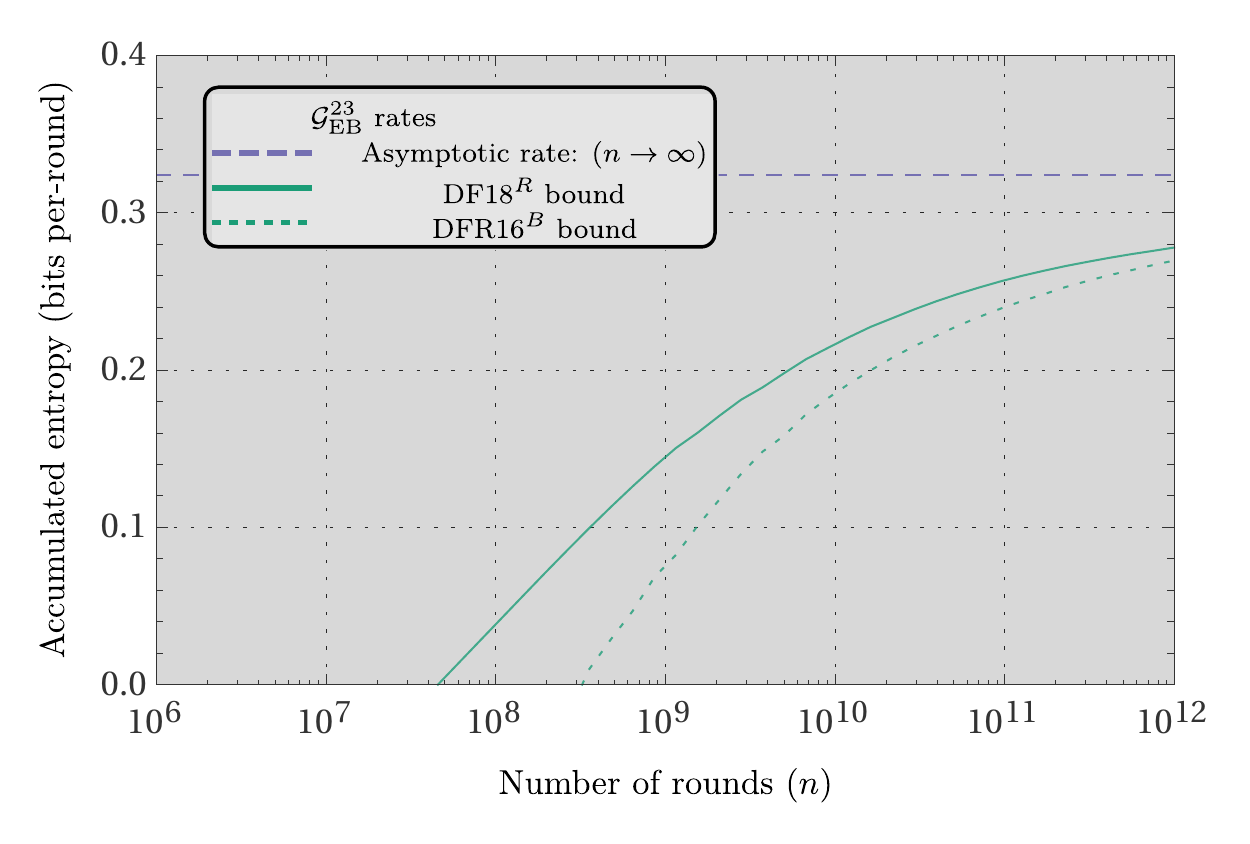}
	\end{minipage}
\end{center}
\caption{Comparison of the certifiable accumulation rates using the two different statements of the EAT: ${\DFR}^B$~\cite{ARV} and ${\DF}^R$ \eqref{eq:EAT-total}. The rates were derived using the following procedure. We assumed a qubit implementation of the protocols with a detection efficiency $\eta = 0.9$, optimizing the state and measurement angles in order to maximise the asymptotic rate. Then, for each value of $n$ an optimization of the min-tradeoff function choice was performed -- for the rates calculated using \eqref{eq:EAT-total} we also optimized the $\beta$ parameter at each value of $n$. To ensure that we approached the asymptotic rate as $n$ increased we set $\gamma = \delta_1 = \dots = \delta_{|\G|} = n^{-1/3}$ as such a choice provides a constant completeness error across all values of $n$.
} \label{fig:EAT-RD-comparison}

\end{figure}

\section{Conic program duality}\label{app:cones}
In this section we outline the duality statements for conic programs,
introduce the alternative form of dual that we use in this paper and
show that it has the required properties to be considered a dual.
\begin{definition}
	A \emph{cone} is a set $\cK\subseteq\mathbb{R}^n$ with the property
	that if $x\in\cK$ then $\lambda x\in\cK$ for all $\lambda\geq0$.  A
	cone is \emph{pointed} if $\cK\cap(-\cK)=\emptyset$.
\end{definition}

\begin{definition}
	Given a cone $\cK$, its \emph{dual cone} is the set
	$\cK^*\subseteq\mathbb{R}^n$ defined by the property
	that $y\in\cK^*$ if and only if $\langle y,x\rangle\geq0$ for all
	$x\in\cK$, i.e., $\cK^*=\{y:\langle y,x\rangle\geq0\
	\forall x\in\cK\}$.
\end{definition}

\begin{definition}
	A \emph{proper cone} is a cone that is closed, convex, pointed and non-empty.
\end{definition}

\begin{definition}[Dual for conic programs]
	Let $\cK\subseteq\mathbb{R}^n$ be a proper cone with dual
	$\cK^*$, $M\in\mathbb{R}^{m\times n}$, $b\in\mathbb{R}^m$ and
	consider the conic program
	$$\min_{x\in\mathbb{R}^n}\langle c,x\rangle\quad\text{subj.\ to }\quad Mx=b,\ x\in\cK\,.$$
	
	The optimization
	$$\max_{y\in\mathbb{R}^m,z\in\mathbb{R}^n}\langle b,y\rangle\quad\text{subj.\ to }\quad c=z+M^T y,\ z\in\cK^*$$
	is the dual program.
\end{definition}
Note that this is a conic program over $\cK^*$, the dual cone to
$\cK$.

The following two Lemmas are standard results (see, for example~\cite{BoydVan})

\begin{lemma}[Weak duality]
	Let $P$ be a conic program with optimum value $p^*$. If the program
	$D$, dual to $P$ has optimum $d^*$ then $p^*\geq d^*$.
\end{lemma}

\begin{lemma}[Strong duality]
	Let $P$ be a conic program with optimum value $p^*$ and dual $D$.
	If $P$ is strictly feasible then $p^*=d^*$.
\end{lemma}

Consider a family of conic programs parameterized by $b$, denoted
$P(b)$, with optimum $p^*(b)$.  Say that $b$ is valid if there exists
some $x\in\cK$ such that $Mx=b$, and denote by $\cB$ the set of valid
$b$.  Consider now the program $\tilde{D}(b)$ defined by
$$\max_{y\in\mathbb{R}^m}\langle b,y\rangle\ \ \ \text{subj.\ to
}\ \ p^*(b')\geq\langle y,b'\rangle\ \ \forall b'\in\cB$$

\begin{lemma}
	If $P(b)$ has optimum $p^*(b)$ and $\tilde{D}(b)$ has optimum
	$\tilde{d}^*(b)$, then $\tilde{d}^*(b)\leq p^*(b)$.  Furthermore, if
	$P(b)$ is strictly feasible, then $\tilde{d}^*(b)=p^*(b)=d^*(b)$.
\end{lemma}
\begin{proof}
	For the first part, note that the set of constraints in $\tilde{D}$ include
	$p^*(b)\geq\langle y,b\rangle$, so $\tilde{d}^*(b)\leq p^*(b)$.
	
	For the second part, consider the dual problem $D(b)$ and write the
	constraint as $c-M^Ty\in\cK^*$. Take the inner product of $c-M^Ty$
	with $x^*(b')$ (the optimal argument for the primal with parameter
	$b'\in\cB$) to give
	$\langle c,x^*(b')\rangle - \langle M^Ty,x^*\rangle=p^*(b')-\langle
	y,Mx^*\rangle=p^*(b')-\langle y,b'\rangle$.
	Since $c-M^Ty\in\cK^*$, from $x^*(b')\in\cK$ we have that
	$p^*(b')-\langle y,b'\rangle\geq 0$.  Thus, for any $b'\in\cB$ we have
	$\langle y^*(b),b'\rangle\leq p^*(b')$.  The constraints in $D$ thus
	imply those in $\tilde{D}$ and so $\tilde{d}^*(b)\geq d^*(b)$.  If
	$P(b)$ is strictly feasible then by strong duality, $p^*(b)=d^*(b)$,
	so, combining with the first part, $\tilde{d}^*(b)=p^*(b)=d^*(b)$.
\end{proof}
\begin{remark}
	The previous lemma implies that we can think of $\tilde{D}$ as an
	alternative dual to $P$.
\end{remark}

\end{document}